\documentclass[aps,prapplied,superscriptaddress,twocolumn,longbibliography]{revtex4-2}
\usepackage{graphicx, color, graphpap}
\usepackage{enumitem}
\usepackage{amssymb}
\usepackage{amsthm}
\usepackage{dsfont}
\usepackage{mathtools}
\usepackage[dvipsnames]{xcolor}
\usepackage{multirow}
\usepackage[colorlinks=true,citecolor=blue,linkcolor=magenta]{hyperref}
\usepackage[T1]{fontenc}
\usepackage{bbm}
\usepackage{thmtools,thm-restate}
\usepackage{verbatim}
\usepackage{titlesec}
\usepackage{amsmath}
\usepackage{graphicx}
\usepackage[lined,boxed,ruled]{algorithm2e}
\usepackage{float}
\usepackage{hyperref}
\usepackage{url}
\usepackage{subfigure}
\usepackage{amsmath}
\usepackage{amssymb}
\usepackage{mathtools}
\usepackage{amsthm}
\usepackage{bm}
\usepackage{xcolor}
\usepackage{color}
\usepackage{textcomp}
\newtheorem{theorem}{Theorem}
\newtheorem{lemma}{Lemma}
\newtheorem{definition}{Definition}

\usepackage{soul}
\begin{document}

\title{Supervised Learning Guarantee for Quantum AdaBoost }

\author{Yabo Wang}
\affiliation{
	Key Laboratory of Systems and Control, Academy of Mathematics and Systems Science, Chinese Academy of Sciences, Beijing 100190, P. R. China
}
\affiliation{
	University of Chinese Academy of Sciences, Beijing 100049, P. R. China
}
\author{Xin Wang}
\affiliation{
 Department of Automation, Tsinghua University, Beijing 100084, P. R. China
}
\author{Bo Qi}
\email{qibo@amss.ac.cn}
\affiliation{
	Key Laboratory of Systems and Control, Academy of Mathematics and Systems Science, Chinese Academy of Sciences, Beijing 100190, P. R. China
}
\affiliation{
	University of Chinese Academy of Sciences, Beijing 100049, P. R. China
}

\author{Daoyi Dong}
\affiliation{
	School of Engineering, Australian National University, Canberra ACT 2601, Australia
}

\date{\today}

\begin{abstract}
In the noisy intermediate-scale quantum (NISQ) era, the capabilities of  variational quantum algorithms are greatly constrained due to a limited number of qubits and shallow depth of quantum circuits. We may view these  variational quantum algorithms as weak learners in supervised learning. Ensemble methods are general approaches to combining weak learners to construct a strong one in machine learning. In this paper, by focusing on classification, we theoretically establish and numerically verify a learning guarantee for quantum adaptive boosting (AdaBoost). The supervised learning risk bound describes how the prediction error of quantum AdaBoost on binary classification decreases as the number of boosting rounds and sample size increase.  We further empirically demonstrate the advantages of quantum AdaBoost by focusing on a 4-class classification.  The quantum AdaBoost not only outperforms several other ensemble methods, but also in the presence of noise, it can surpass the ideally noiseless but unboosted primitive classifier   after only a few boosting rounds. Our work indicates that in the current NISQ era, introducing appropriate ensemble methods is particularly valuable in improving the performance of quantum machine learning algorithms.
\end{abstract}

\maketitle
\section{Introduction}

Machine learning has achieved remarkable success in various fields with a wide range of applications \cite{krenn2023artificial,xiao2023practical,lu2018separability,luo2023detecting,shingu2021boltzmann}. A major objective of machine learning is to develop efficient and accurate prediction algorithms, even for large-scale problems \cite{mohri2018foundations,wright2022high}. The figure of merit, prediction error, can be decomposed into the summation of training  and generalization errors. Both of them should be made small to guarantee an accurate prediction. However, there is a tradeoff between reducing the training error and restricting the generalization error through controlling the size of the hypothesis set, known as Occam’s Razor principle \cite{rasmussen2000occam}.
	
	For classical machine learning, empirical studies have demonstrated that the training error can often be effectively minimized  despite the non-convex nature and abundance of spurious minima in training loss landscapes \cite{li2018visualizing,geiger2021landscape}. This observation has been explained by the theory of over-parametrization  \cite{rocks2022memorizing,baldassi2022learning,allen2019convergence,larocca2023theory}. However, it is still  difficult to theoretically describe how to guarantee a good generalization, which is one of the key problems to be solved in classical machine learning.
	
	Owing to the immense  potential of quantum computing, extensive efforts have been dedicated to developing quantum machine learning \cite{biamonte2017quantum,beer2023quantum,uvarov2020machine,banchi2021generalization,cerezo2022challenges}. However, in the noisy intermediate-scale quantum (NISQ) era, the capability of quantum machine learning is greatly constrained due to a  limited number of qubits and shallow depth of the involved quantum circuits. Algorithms based on parameterized quantum circuits (PQCs) have become the leading candidates to yield potential quantum advantages in the era of NISQ \cite{benedetti2019parameterized}. The basic idea behind them is that these parameterized quantum models can provide  representational and/or computational powers beyond what is possible with classical models \cite{abbas2021power,schuld2021effect,liu2021rigorous,Boris2023,huang2022quantum}.

There are mainly three kinds of parameterized quantum models \cite{jerbi2023quantum}: (a) explicit models \cite{Cerezo2021v,Benedetti2019}, where data are first encoded into quantum states, after undergoing a PQC, the quantum states are measured and the information is used to update the variational parameters through a classical routine; (b) implicit kernel models \cite{havlivcek2019,schuld2019quantum}, where the kernel matrices of the encoding data are computed through quantum circuits, and then used to label data; (c) re-uploading models \cite{perez2020data}, where encoding and parameterized circuits are interleaved. A unified framework has been set for the three quantum models, and it was pointed out that the advantages of quantum machine learning may lie beyond kernel methods~\cite{jerbi2023quantum}. They found that although kernel methods are guaranteed to achieve a lower training error, their generalization power is poor. Thus, both the training and generalization errors should be taken into account when evaluating the prediction accuracy~\cite{wang2024power}.

	It has been proved that good generalization can be guaranteed from few training data for explicit quantum machine learning models~\cite{caro2022gen}. However, in contrast to the classical case, training quantum models is notoriously difficult. This is because when training a PQC, it often suffers from the phenomenon of a barren plateau~\cite{mcclean2018barren,haug2021capacity,Cerezo2021c,marrero2021entanglement,zhao2021analyzing,wang2023trainability,wang2023eha}, where the cost gradient vanishes exponentially fast. Another bottleneck for training a PQC is that the landscape of  the cost function may be swamped with  (exponentially) many spurious local minima~\cite{anschuetz2022quantum,you2021exponentially}. In this sense, most quantum learning algorithms can be viewed as weak learners in the language of supervised machine learning.
	
	To improve the performance of variational quantum algorithms, we can employ ensemble methods as inspired by the classical ensemble learning. There are various kinds of ensemble methods, e.g., bagging \cite{breiman1996bagging}, plurality voting \cite{lam1997application,lin2003performance} and boosting \cite{schapire1999brief}. It has been suggested that an optimized weighted mixture of concepts (learners), e.g., PAC-Bayesian~\cite{mcallester1999pac}, is a promising candidate for further research~\cite{jiang2019fantastic}. Thus, adaptive boosting (AdaBoost) is particularly suitable for improving the performance of weak learners. It sequentially trains a series of base learners with adaptively adjusted cost functions, and then combines them properly to obtain a strong learner, which has higher prediction accuracy compared to  base learners.  For classical ensemble learning, there has been a rich theoretical analysis on AdaBoost \cite{freund1997decision,bartlett1998boosting,gronlund2019margin}, and it has been verified to be effective in practice \cite{sun2021adagcn,drucker1993boosting,li2008adaboost,zhang2019research}.

Various quantum versions of classical AdaBoost have been proposed~\cite{arunachalam2020,wang2021quantum,ohno2022}. In these works, they employed quantum subroutines, e.g., mean estimation and amplitude amplification, to update quantum weak classifiers and estimate the weighted errors to reduce the time complexity. Therefore, the realizations of these quantum versions of AdaBoost are beyond the scope of current NISQ circuits. Recently, ensemble learning has been proposed with NISQ circuits~\cite{Li2023ense,incudini2023}. In Ref.~\cite{Li2023ense}, by focusing on error mitigation, it was empirically shown that both the bootstrap aggregating and the adaptive boosting can significantly enhance the performance of variational quantum classifiers for both classical and quantum datasets, and the adaptive boosting outperforms the bootstrap aggregating. It was pointed out that with ensemble techniques resources can be saved in terms of the number of qubits, gates, and training samples~\cite{incudini2023}.

In this paper, we provide a theoretical learning guarantee for quantum AdaBoost on binary classification,  describing how its prediction error converges to 0 as the number of boosting rounds and sample size increase.  The learning risk bound is also verified via numerical simulations. Then by focusing on a 4-class classification, we further investigate advantages of quantum AdaBoost. In absence of noise, we find that the quantum AdaBoost outperforms some other ensemble methods.  In the presence of weak noise, we demonstrate that the prediction accuracy of quantum AdaBoost is even higher than that of the ideally noiseless but unboosted primitive classifier after only a few boosting rounds. This reinforces the finding in Ref.~\cite{Li2023ense} that ensemble learning can help mitigate the impact of noises which are inherent in NISQ circuits. Our findings indicate that quantum AdaBoost is particularly valuable for improving the performance of quantum machine learning, especially in the NISQ era.

The paper is organized as follows. We  first introduce the quantum classifier and quantum AdaBoost  in Section~\ref{preliminary}. Then we present  our theoretical and empirical results of quantum AdaBoost on binary classification  in Section~\ref{main}. Next, we empirically demonstrate advantages of quantum AdaBoost by focusing on a 4-class classification in Section~\ref{extension}. Section~\ref{conclusion} concludes the paper.

\section{\label{preliminary}Quantum Classifier and Quantum AdaBoost}
	\subsection{Quantum Classifier}
	We start with briefly introducing some quantum notation.
	In quantum computing, information is described in terms of quantum states. For an $N$-qubit system, the quantum state  $\rho$ can be mathematically represented as a positive semi-definite Hermitian matrix $\rho\in \mathbb{C}^{2^N\times 2^N}$ with  $\mathrm{Tr}\left[  \rho\right] =1$. The elementary quantum gates are mathematically described by unitary matrices. A quantum gate $U$ acting on a quantum state $\rho$ takes the
	state to the output state as $U\rho U^{\dag}$ where $U^{\dag}$ is the conjugate and transpose of $U$. When measuring an observable $O$, which is a Hermitian operator, at quantum state $\rho$, its expectation is ${\rm Tr}\left[O\rho \right]$.

	For a $D$-class classification problem, suppose that both  the training and test data are independent and identically distributed~(i.i.d.)~according to some fixed but unknown distribution $\mathcal{D}$ defined over the sample and label space $\mathcal{X}\times\mathcal{Y}$.  When the training set $S=\left\lbrace \left( \boldsymbol{x}_i, y_i\right) \right\rbrace_{i=1}^{n} $ is classical, we may first choose a  quantum encoding circuit  to embed the classical data $\boldsymbol{x}_i$ into quantum state $\rho\left( \boldsymbol{x}_i\right) $ \cite{schuld2021effect,goto2021universal,lloyd2020quantum}. We thus focus on the explicit quantum model~\cite{jerbi2023quantum}. Without loss of generality, we only consider the case where the data are quantum in the following, namely,  $S={\left\lbrace \left( \rho\left( \boldsymbol{x}_i\right) , y_i\right) \right\rbrace} ^n_{i=1}\subset \mathcal{X}\times\mathcal{Y}$.  For a $D$-class classification, denote $\mathcal{Y}=\left\lbrace 1,\cdots,D\right\rbrace \triangleq\left[D \right] $.
	
	To label $\rho(\boldsymbol{x})$, a quantum hypothesis or classifier $h_{\boldsymbol{\theta}}(\cdot)$  can be described in the form of
	\begin{equation}\label{hypothesis}
	h_{\boldsymbol{\theta}}\left(\boldsymbol{x} \right) =\mathop{\arg\max}\limits_{d\in \left[D \right] }{\rm Tr}\left[ \Pi_{d}U\left(\boldsymbol{\theta} \right)\rho\left( \boldsymbol{x}\right)U^{\dagger}\left(\boldsymbol{\theta} \right)\right].
	\end{equation}
	Here,  $\left\lbrace \Pi_d\right\rbrace ^{D}_{d=1}$ are disjoint projectors with $\Pi_d$  relating to the $d$th class for $d\in\left[ D \right] $, and $U\left( \boldsymbol{\theta}\right) $ describes the action of a PQC with $\boldsymbol{\theta}$ being the trainable or variational parameters. For a $D$-class classification, we may measure  $\lceil \log_2 D \rceil$ qubits of the quantum circuits in the computational basis. Here,  $\lceil  \cdot \rceil$ denotes the roundup function, and each computational basis corresponds to a projector representing a particular class.
	
	\begin{figure}[ht]
		\centering
		\includegraphics[width=0.7\linewidth]{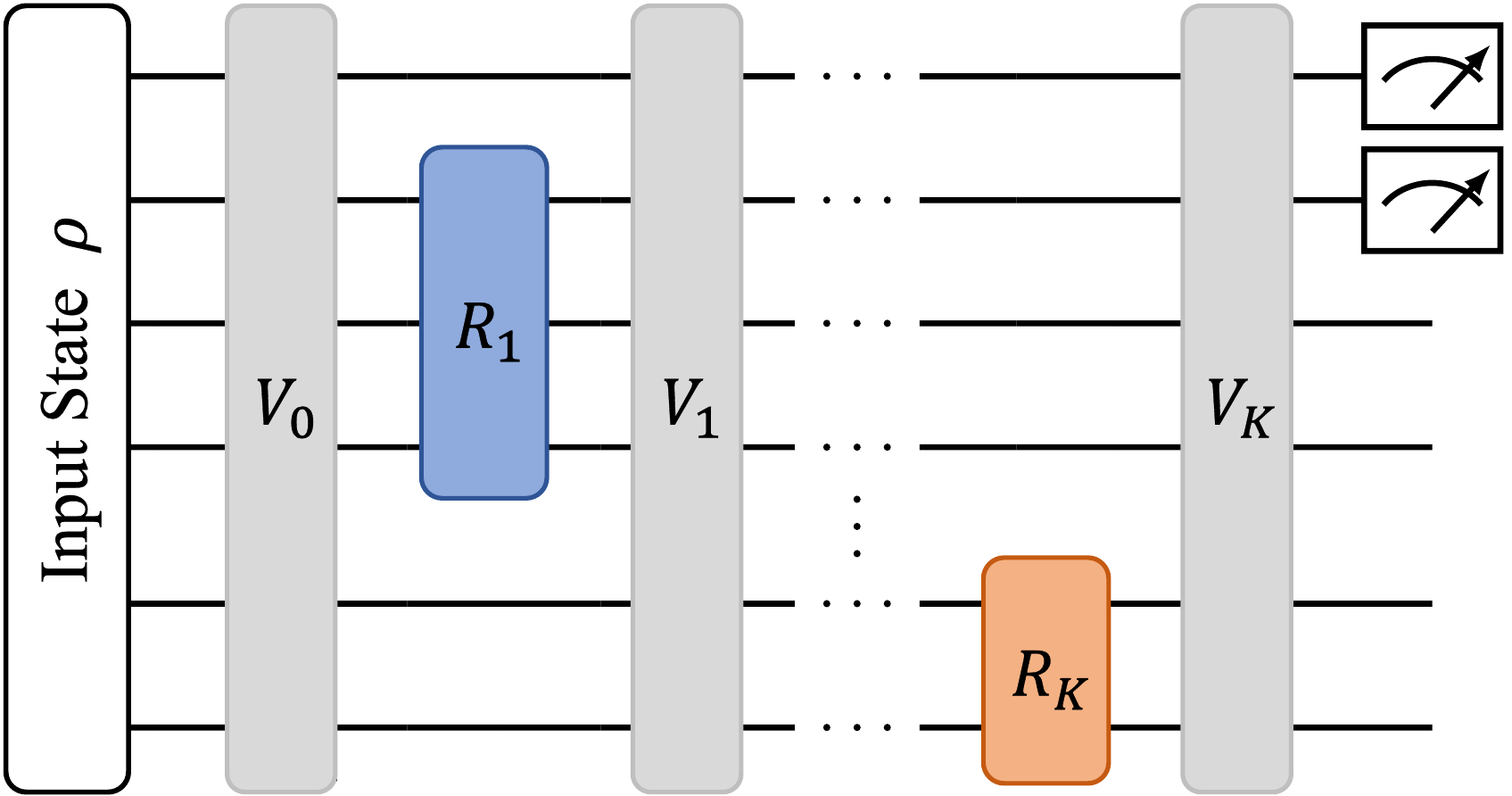}
		\caption{\label{fig:PQC}Schematic of a PQC with $K$ independent and trainable gates. Each trainable gate $R_k$ is  a parameterized multi-qubit rotational gate and can be implemented efficiently.}
	\end{figure}
	
	To be specific, as illustrated in Fig.~\ref{fig:PQC}, suppose that the employed PQC is composed of a total number of $K$ independent parameterized gates and a set of non-trainable gates $\{V_k\}^K_{k=0}$, whose action can be described as
	\begin{equation}\label{pqc}
	U\left(\boldsymbol{\theta} \right) =\prod ^{K}_{k=1}\Big{[}V_k{R_{k}^{\left(i_{k}, j_{k}\right) }\left(\theta_{k} \right) }\Big{]} \cdot V_0,
	\end{equation}
	where $\boldsymbol{\theta}=\left(\theta_{1}, \cdots , \theta_{K} \right) $ denotes a $K$-dimensional parameter vector. For each $k$, the trainable gate $R_{k}^{\left(i_{k}, j_{k}\right) }\left(\theta_{k} \right) $ denotes a rotational gate with angle $\theta_{k}$ around a $j_{k}$-qubit Pauli tensor product operator $P_{k}$, which acts non-trivially on the $i_{k}$th  through to the $\left( i_{k}+j_{k}-1\right) $th qubits, namely,
	\begin{align*}
	&R_{k}^{\left(i_{k}, j_{k}\right) }\left(\theta_{k} \right)
	=I^{\otimes \left( i_{k}-1\right) }\otimes e^{-i\frac{\theta_{k}}{2}P_{k}}\otimes I^{\otimes \left( n-i_{k}-j_{k}+1\right) }.
	\end{align*}
In practice, these multi-qubit rotational gates can be  implemented efficiently by some typical two-qubit  gates and parameterized single-qubit gates. For example, as illustrated in Fig.~\ref{fig:RZZ}, the multi-qubit rotational gates around $z$ axis can be implemented by a single-qubit rotational gate around $z$ axis and some two-qubit CNOT gates.
	\begin{figure}[t]
		\centering
		\includegraphics[width=0.35\textwidth]{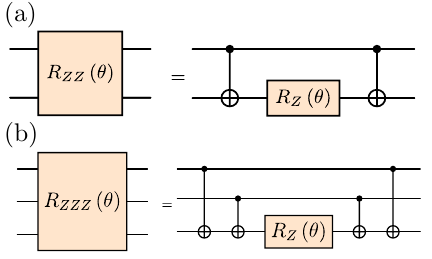}
		\caption{\label{fig:RZZ} Hardware-efficient implementations of multi-qubit rotational gates. (a) The module of $2$-qubit rotational gate $R_{ZZ}\left( \theta\right) $ around  $Z\otimes Z$. (b) The module of $3$-qubit rotational gate $R_{ZZZ}\left( \theta\right) $ around  $Z\otimes Z\otimes Z$.}
	\end{figure}

	The prediction error or expected risk of the quantum classifier $h_{\boldsymbol{\theta}}$ is defined as
	\begin{equation}\label{prediction}
	R\left( h_{\boldsymbol{\theta}}\right)=\mathop{\mathbb{E}}\limits_{\left(\boldsymbol{x},y\right) \sim\mathcal{D}}{\mathbb{I}_{ h_{\boldsymbol{\theta}}\left( \boldsymbol{x}\right) \neq y}} =\mathop{\mathbb{P}}\limits_{\left(\boldsymbol{x},y\right)\sim\mathcal{D}}{\left[ \,h_{\boldsymbol{\theta}}\left( \boldsymbol{x}\right)\neq y \,\right] } .
	\end{equation}
Here, $\mathbb{I}_{\mathcal{A}}$ denotes the indicator function of an event $\mathcal{A}$. If  $\mathcal{A}$ is true, $\mathbb{I}_{\mathcal{A}}=1$, and $\mathbb{I}_{\mathcal{A}}=0$, otherwise.
	
	The prediction error of a hypothesis is not directly accessible, since both the label of unseen data and the distribution $\mathcal{D}$ are unavailable. However, we  can take the training error or empirical risk of $h_{\boldsymbol{\theta}}$ as a proxy, defined as
	\begin{equation}\label{training}
	\widehat{R}_{S}\left( h_{\boldsymbol{\theta}}\right)=\frac{1}{n}\sum_{i=1}^{n}{\mathbb{I}_{h_{\boldsymbol{\theta}}\left( \boldsymbol{x}_{i}\right)\neq y_{i}}}.
	\end{equation}
	The difference between the prediction error $R\left( h_{\boldsymbol{\theta}}\right)
	$ and the training error $\widehat{R}_{S}\left( h_{\boldsymbol{\theta}}\right)$ is
	referred to as the generalization error, which reads
	\begin{equation}\label{generalization}
	\texttt{gen}\left( h_{\boldsymbol{\theta}}\right)=R\left( h_{\boldsymbol{\theta}}\right)-\widehat{R}_{S}\left( h_{\boldsymbol{\theta}}\right).
	\end{equation}
	It is clear that to make accurate predictions, both the training and generalization errors should be small.

	\subsection{Quantum AdaBoost}
	
Recall that for a $D$-class classification, the projectors $\{\Pi_d\}^D_{d=1}$ may be selected to correspond to the computational basis of  $\lceil \log_2 D \rceil$ qubits of the quantum circuits. Usually, we can train a quantum classifier $h_{\boldsymbol{\theta}}$ described by Eq.~(\ref{hypothesis})  utilizing the cross entropy function
\begin{equation}\label{loss0}
	{\mathcal{L}_0\left( \boldsymbol{\theta};S\right)  = - \sum_{i=1}^{n} \frac{1}{n} \boldsymbol{y}_i^{\top} \log\left( \boldsymbol{p}_i(\boldsymbol{\theta})\right) }.
	\end{equation}
Here, each label $y_i\in \left[D \right]$  has been transformed into a $D$-dimensional one-hot vector denoted by $\boldsymbol{y}_i$, and $\boldsymbol{p}_i (\boldsymbol{\theta}) = \left[ p_{i,1}(\boldsymbol{\theta}) ,\cdots,p_{i,D }(\boldsymbol{\theta} )\right] ^{\top}$ with $$p_{i,d}(\boldsymbol{\theta} ) = \operatorname{Tr}\left[ \Pi_{d}U\left(\boldsymbol{\theta} \right)\rho\left( \boldsymbol{x}_{i}\right)U^{\dagger}\left(\boldsymbol{\theta} \right)\right]$$  for each $d \in \left[D \right]$.
Note that in Eq.~(\ref{loss0}), the  logarithmic function $\log(\cdot)$ acts on each element of the vector
$\boldsymbol{p}_i (\boldsymbol{\theta})$, and the weights of all training data are set to be equal.

Intuitively, a quantum  classifier is considered weak if its training error is only slightly less than that of a random guess, which is $\frac{D-1}{D}$. Conversely, it is considered strong if its training error is only slightly greater than 0. A more precise definition can be found in Ref.~\cite{mohri2018foundations}. In the NISQ era, due to shallow depth, a limited number of qubits, and inevitable noises in quantum circuits, quantum classifiers are generally weak.

Inspired by the classical multi-class AdaBoost \cite{hastie2009}, we now introduce the quantum AdaBoost. It sequentially trains a series of weak quantum classifiers by adaptively adjusting the weights of the data in the cross entropy function, and then combines them properly to yield a strong quantum classifier.

Assume that the quantum AdaBoost consists of $T$ boosting rounds in total. In practice,  $T$ is  selected via cross-validation. Denote by $h_{\boldsymbol{\theta}_t}$ the weak quantum base classifier at the $t$th round, where $\boldsymbol{\theta}_t=\left(\theta_{1,t}, \cdots , \theta_{K,t} \right) $ represent the correponding variational parameters of the employed PQC.
To find  $h_{\boldsymbol{\theta}_t}$, we can utilize the following cross-entropy loss function:
	\begin{equation}\label{loss1}
	{\mathcal{L}_t\left( \boldsymbol{\theta};S\right)  = - \sum_{i=1}^{n} \mathcal{D}_{t}\left(i \right)  \boldsymbol{y}_i^{\top} \log\left( \boldsymbol{p}_i(\boldsymbol{\theta})\right) }.
	\end{equation}
Here, $\mathcal{D}_t(i)$ denotes the distribution over the training data  $S=\{(\rho\left( \boldsymbol{x}_i\right) , y_i)\}^n_{i=1}$ at round $t$. Initially, $\mathcal{D}_1(i)=\frac{1}{n}$, namely,  the weights of all training data  are initially set equal as that in Eq.~(\ref{loss0}).  The corresponding weighted error $\epsilon_t$ of $h_{\boldsymbol{\theta}_t}$ on the training data weighted by the distribution $\mathcal{D}_t$ reads
	\begin{equation}\label{weighted0}
	\epsilon_t=\sum_{i=1}^{n}{\mathcal{D}_{t}\left(i \right)\mathbb{I}_{h_{\boldsymbol{\theta}_t}\left( \boldsymbol{x}_{i}\right) \neq y_{i} }}.
	\end{equation}

Once $h_{\boldsymbol{\theta}_t}$ is chosen, we can update the distribution over data as $$\mathcal{D}_{t+1}\left(i \right)\propto\mathcal{D}_{t}\left(i \right)\exp{\left[ \alpha_{t}\mathbb{I}_{y_{i}\neq h_{\boldsymbol{\theta}_t}\left( \boldsymbol{x}_{i}\right) }\right] },$$ where $\alpha_t=\log{\frac{1-\epsilon_{t}}{\epsilon_{t}}}+\log{\left( D-1\right) }$ is nonnegative. The basic idea is
increasing the weight on the sample $(\rho\left( \boldsymbol{x}_i\right) , y_i)$ if it is incorrectly classified by $h_{\boldsymbol{\theta}_t}$, namely, $h_{\boldsymbol{\theta}_t}(\boldsymbol{x}_i)\neq y_i$,  and, on the contrary, decreasing the weight if $(\rho\left( \boldsymbol{x}_i\right), y_i)$ is correctly classified. Thus, in the next round of boosting, the quantum classifier can focus more on the points incorrectly classified.

After $T$ rounds of boosting, the classifier returned by quantum AdaBoost takes the form of $$h= \mathop{\arg\max}\limits_{d\in \left[D \right] }{\sum_{t=1}^{T}{\alpha_{t}\mathbb{I}_{h_{\boldsymbol{\theta}_t}=d}}},$$
with the weight $\alpha_t$ for the $t$th weak base classifier $h_{\boldsymbol{\theta}_t}$.
	
We summarize the procedures of $D$-class quantum AdaBoost in Algorithm~\ref{algorithm}. Note that the main procedures are similar to those in Ref.~\cite{Li2023ense}. When training a base classifier, we can use other loss functions besides the cross-entropy loss function. The key is to make sure that the obtained base classifier should surpass the random guess, namely, its weighted error is less than $\frac{D-1}{D}$.
	
	\begin{algorithm}
		\SetKwInOut{Input}{Input}\SetKwInOut{Output}{Output}
		\Input{Training set  $S={\left\lbrace \left( \rho\left( \boldsymbol{x}_i\right) ,y_i\right) \right\rbrace}^n_{i=1}$
			\\
			Number of boosting rounds $T$\\
			Distribution $\mathcal{D}_{1}\left(i \right) = \frac{1}{n}$, for $i \in\left[n \right]$
		}
		\BlankLine
		\For{$t\leftarrow 1$ \KwTo $T$}{
		Train a weak quantum classifier $h_{\boldsymbol{\theta}_t}$ with $\epsilon_{t}<\frac{D-1}{D}$
			\BlankLine
			$\alpha_{t}\leftarrow\log{\frac{1-\epsilon_{t}}{\epsilon_{t}}}+\log{\left( D-1\right) }$
			\BlankLine
			\For{$i\leftarrow 1$ \KwTo $n$}{
				$\mathcal{D}_{t+1}\left(i \right)\leftarrow \mathcal{D}_{t}\left(i \right)\exp{\left[ \alpha_{t}\mathbb{I}_{y_{i}\neq h_{\boldsymbol{\theta}_t}\left( \boldsymbol{x}_{i}\right) }\right] }$
			}
			\BlankLine
			Normalize $\left\lbrace\mathcal{D}_{t+1}\left(i \right)\right\rbrace_{i=1}^{n}$
			
		}
		$h\leftarrow \mathop{\arg\max}\limits_{d\in \left[D \right] }{\sum_{t=1}^{T}{\alpha_{t}\mathbb{I}_{h_{\boldsymbol{\theta}_t}=d}}}$
		\BlankLine
		\Output{Predictor $h$}
		\caption{$D$-Class Quantum AdaBoost}\label{algorithm}
	\end{algorithm}

	\section{\label{main} Learning Guarantee for Quantum AdaBoost on Binary Classification}
	In this section, by focusing on binary classification, we establish and verify a prediction risk bound for quantum AdaBoost.
	
	\subsection{Quantum AdaBoost on Binary Classification}
	
	For multi-class classification, an alternative approach is to reduce the problem to multiple binary classification tasks. For each task, a binary classifier is returned, and the multi-class classifier is defined by a combination of these binary classifiers. Two standard reduction techniques are one-versus-the-rest and one-versus-one \cite{mohri2018foundations}. In this subsection, we focus on the basic binary  quantum AdaBoost.
	
	For  binary classification, it is more convenient to denote the label space by $\mathcal{Y}=\left\lbrace -1, +1\right\rbrace $. The  quantum classifier $h_{\boldsymbol{\theta}}\left( \cdot\right) $ can be defined in terms of the Pauli-$Z$ operator $Z=\begin{pmatrix}
	1 & 0\\
	0 & -1
	\end{pmatrix}$ as
	\begin{equation}\label{hypothesis2}
	h_{\boldsymbol{\theta}}\left(\boldsymbol{x} \right) ={\rm sgn}\left({\rm Tr}\left[ Z U\left(\boldsymbol{\theta} \right)\rho\left( \boldsymbol{x}\right)U^{\dagger}\left(\boldsymbol{\theta} \right)\right]\right),
	\end{equation}
where the sign function $\text{sgn}\left( x \right)$, equals $1$ if $x\textgreater \,0$, $-1$ if $x\textless \,0$, and 0 otherwise.
	
	It is straightforward to verify that for a sample $\left( \rho\left( \boldsymbol{x}\right), y\right) $, the following  important relation holds:
	\begin{equation}\label{equality}
	\mathbb{I}_{y\neq h_{\boldsymbol{\theta}}\left( \boldsymbol{x}\right)}
	=\mathbb{I}_{yh_{\boldsymbol{\theta}}\left( \boldsymbol{x}\right)\le 0 }.
	\end{equation}
By employing Eq.~(\ref{equality}) and inspired by the classical binary AdaBoost \cite{mohri2018foundations}, we can modify Algorithm~\ref{algorithm} slightly  to make it more suitable for binary classification as presented in Algorithm~\ref{algorithm2}.
	
\begin{algorithm}
		\SetKwInOut{Input}{input}\SetKwInOut{Output}{output}
		\Input{Training set  $S={\left\lbrace \left( \rho\left( \boldsymbol{x}_i\right) , y_i\right) \right\rbrace}^n_{i=1}$
			\\
			Number of boosting rounds $T$\\
			Distribution $\mathcal{D}_{1}\left(i \right) = \frac{1}{n}$, for $i \in\left[n \right]$
		}
		\BlankLine
		\For{$t\leftarrow 1$ \KwTo $T$}{
			Train a weak quantum classifier $h_{\boldsymbol{\theta}_t}$  with  $\epsilon_{t}<\frac{1}{2}$
			\BlankLine
			$\alpha_{t}\leftarrow\frac{1}{2}\log{\frac{1-\epsilon_{t}}{\epsilon_{t}}}$
			\BlankLine
			$Z_{t} \leftarrow 2{\left[ \epsilon_{t}\left(1-\epsilon_{t} \right) \right]}^{\frac{1}{2}} \ $ ~~  //normalization factor
			\BlankLine
			\For{$i\leftarrow 1$ \KwTo $n$}{
				$\mathcal{D}_{t+1}\left(i \right)\leftarrow \frac{\mathcal{D}_{t}\left(i \right)\exp{\left[ -\alpha_{t}y_{i}h_{\boldsymbol{\theta}_t}\left(x_{i} \right)\right]  }}{Z_{t}}$
			}
		}
		$h\leftarrow \text{sgn}{\left( \sum_{t=1}^{T}{\alpha_{t}h_{\boldsymbol{\theta}_t}}\right) }$
		\BlankLine
		\Output{Predictor $h$}
		\caption{Binary Quantum AdaBoost}\label{algorithm2}
	\end{algorithm}
	
	Different from Algorithm~\ref{algorithm}, the quantum classifier $h_{\boldsymbol{\theta}_t}$  in Algorithm~\ref{algorithm2}  is described by Eq.~(\ref{hypothesis2}). The distribution over data  updates according to
	$$\mathcal{D}_{t+1}\left(i \right)=\frac{\mathcal{D}_{t}\left(i \right)\exp{\left[ -\alpha_{t}y_{i}h_{\boldsymbol{\theta}_t}\left(x_{i} \right)\right]  }}{Z_{t}}.$$
The basic idea is the same as that in Algorithm~\ref{algorithm}, namely, letting the quantum classifier in the next round focus more on the points incorrectly classified in the current round.
It can be verified that the weight $\alpha_t=\frac{1}{2}\log \frac{1-\epsilon_t}{\epsilon_t}$  is chosen to minimize the upper bound of the empirical risk $\widehat{R}_{S}\left( h\right) $ of the binary quantum AdaBoost~\cite{mohri2018foundations}.

	\subsection{Learning Guarantee}

	We establish the learning risk bound for  binary quantum AdaBoost in the following theorem, whose proof can be found in Appendix~\ref{proofth01}.
	
	\begin{theorem}\label{th01}
		For the binary  quantum AdaBoost Algorithm~\ref{algorithm2}, assume that the employed PQC has a total number of $K$ independent parameterized gates. Then for  $0<\delta<1$, with a probability  at least $1-\delta$ over the draw of an i.i.d. $n$-size sample set,
		the prediction error $R\left( h \right)$ of the returned binary  quantum AdaBoost classifier satisfies
		\begin{equation}
  R\left( h \right)\le e^{-2\sum_{t=1}^{T}(\frac{1}{2}-\epsilon_t)^2} + 12\sqrt{\frac{K\log 7K}{n}}+4\sqrt{\frac{K}{n}}+\sqrt{\frac{\log{\frac{1}{\delta}} }{2n}}\label{th01b}.
		\end{equation}
\end{theorem}
	
 Recall that the prediction error $R(h)$ can be decomposed into the training error $\widehat{R}_S\left( h\right)$ and generalization error $\texttt{gen}\left( h\right)$. The first term in the RHS of Eq.~(\ref{th01b}) gives the upper bound of the training error $\widehat{R}_S\left( h\right)$. It decreases exponentially as the number of boosting $T$ increases. This is attributed to the good property of AdaBoost.

 The second and third terms in the RHS of Eq.~(\ref{th01b}) can be viewed as  the upper bound of the generalization error $\texttt{gen}\left( h\right)$. Here, in contrast to the classical case, our success in bounding the generalization error of quantum AdaBoost owes to the good generalization property of quantum machine learning~\cite{caro2022gen}. We find that as the number of  trainable gates $K$ increases, on the one hand, the second and third terms become larger implying the increase of the generalization error; on the other hand, the expressibility of PQC  will also be larger implying that each weighted error $\epsilon_t$ can be reduced leading to the reduction of  the first term in the RHS of Eq.~(\ref{th01b}). On the contrary, as $K$ decreases, the second and third terms decrease, while the decrease in the expressibility of PQC cause the increase of $\epsilon_t$, which in turn results in an increase of  the first term. Therefore, there is a tradeoff between reducing the training error and restricting the generalization error through controlling the number of trainable gates in the PQC, which is  Occam’s Razor principle.

The last term in the RHS of Eq.~(\ref{th01b})  is due to McDiarmid's concentration inequality~\cite{mohri2018foundations}, which gives probability bounds for a random variable to deviate from its mean.

It is important to take into account of the effect of noises inherent in NISQ circuits.  From Theorem~\ref{th01}, we can see that for noisy PQCs,  the learning risk bound  Eq.~(\ref{th01b}) is still valid as long as a weak quantum classifier, which is superior to the random guess, can be trained at each boosting round.  However, it is worth pointing out that this assumption will become hard to be met as the noise rate gets larger.

	\subsection{Numerical Verification}
In this subsection we verify the learning risk bound given in Theorem~\ref{th01} by focusing on a quantum phase classification problem. The numerical simulations are conducted using the Python libraries PennyLane \cite{pennylane2018} and JAX \cite{jax2018github}. Specifically, we utilize the PennyLane's built-in `default.qubit' device to simulate the circuit. JAX is employed to calculate the loss function and its gradients with respect to the circuit parameters.

We consider the axial next-nearest-neighbor Ising (ANNNI) model consisting of $N$ spins \cite{Monaco2023,cea2024,Beccaria2007}
\begin{equation}\label{H}
\begin{aligned}
  H=-J\left( \sum_{i=1}^{N-1} X_{i} X_{i+1}-\kappa\sum^{N-2}_{i=1}    X_{i} X_{i+2}+h\sum^N_{i=1}  Z_i\right),
\end{aligned}
\end{equation}
where $X_i=\begin{pmatrix}
	0 & 1\\
	1 & 0
	\end{pmatrix}$ and $Z_i=\begin{pmatrix}
	1 & 0\\
	0 & -1
	\end{pmatrix}$ are Pauli matrices acting on the $i$th spin, and we assume open boundary conditions. The constant $J$ describes the energy scale of the Hamiltonian and is set to be $J=1$ without loss of generality. The adimensional ratios $\kappa$ and $h$ account for the next-nearest-neighbor interaction and the transverse magnetic field, respectively. The difference of sign between the nearest and next-nearest interactions is responsible for the magnetic frustration, and $\kappa$ is referred to as the frustration parameter.

This is a paradigm to study the competition between magnetic ordering, frustration, and thermal disordering effects. Owing to the competition, the phase diagram of the ground states of $H$ is rich. For the sake of simplicity, as in Ref.~\cite{Monaco2023}, we only consider the three phases that have been confirmed, namely,  paramagnetic, ferromagnetic, and antiphase. When the frustration is low ($\kappa<0.5$), the critical Ising-like line $h_I(\kappa) \approx \frac{1-\kappa}{\kappa} \ \left(1-\sqrt{\frac{1-3 \kappa+4 \kappa^{2}}{1-\kappa}}\right) $ separates the paramagnetic and ferromagnetic phases. For high frustration ($\kappa>0.5$), of a commensurate-incommensurate type, the critical line $h_{C}(\kappa) \approx 1.05 \sqrt{(\kappa-0.5)(\kappa-0.1)}$ separates the paramagnetic and antiphase. More details can be found in Refs.~\cite{Monaco2023,cea2024}. Since the paramagnetic phase is disordered, while both the ferromagnetic and antiphase are ordered, we consider binary classification. For ground states of $H$, we label them as $+1$ if they belong to the paramagnetic phase; otherwise, we label $-1$, namely, they either have a ferromagnetic or an antiphase.
	
Consider a $6$-qubit ANNNI model. We draw $\left( \kappa,h \right) \in (0,1) \times (0,2)$ according to the uniform distribution and then compute the ground state of the corresponding Hamiltonian $H(\kappa, h)$ of Eq.~(\ref{H}). The ground state and its label of the quantum phase constitute a sample point. Repeat the above procedure $n$ times, and we can generate a data set $S={\left\lbrace \left( \rho\left(\kappa_i, h_i\right) , y_i\right) \right\rbrace} ^n_{i=1}$. Here, without loss of generality, we assume that the sample points in $S$ are different from each other.

In this paper, we employ quantum convolutional neural networks (QCNNs) as our PQCs.  QCNNs have been widely used as quantum classifiers \cite{wei2022quantum, chen2022quantum, hur2022quantum}. It has been proved that barren plateaus can be avoided by using QCNNs~\cite{pesah2021absence}, and thus the trainability can be effectively improved. In the following, to determine the quantum classifier $h_{\boldsymbol{\theta}_t}$ at the $t$th boosting round, we first initialize the variational parameters of the employed QCNN according to the standard normal distribution, and then update them by the Adam \cite{kingma2014adam}.  The learning rate is set to be $0.05$, the gradient of the cross-entropy loss (Eq.~(\ref{loss1})) is calculated via the parameter-shift-rule~\cite{Romero2018,Mitarai2018,schuld2019evaluating}, and the maximum number of iterations is set to be $120$. Since  after each iteration  a quantum classifier is yielded, we adopt the  classifier having the minimum training error $\epsilon_t$ among the $120$ iterations as the weak quantum base classifier  $h_{\boldsymbol{\theta}_t}$ at the  $t$th boosting round.

\begin{figure}[h]
	\centering
	\includegraphics[width=\linewidth]{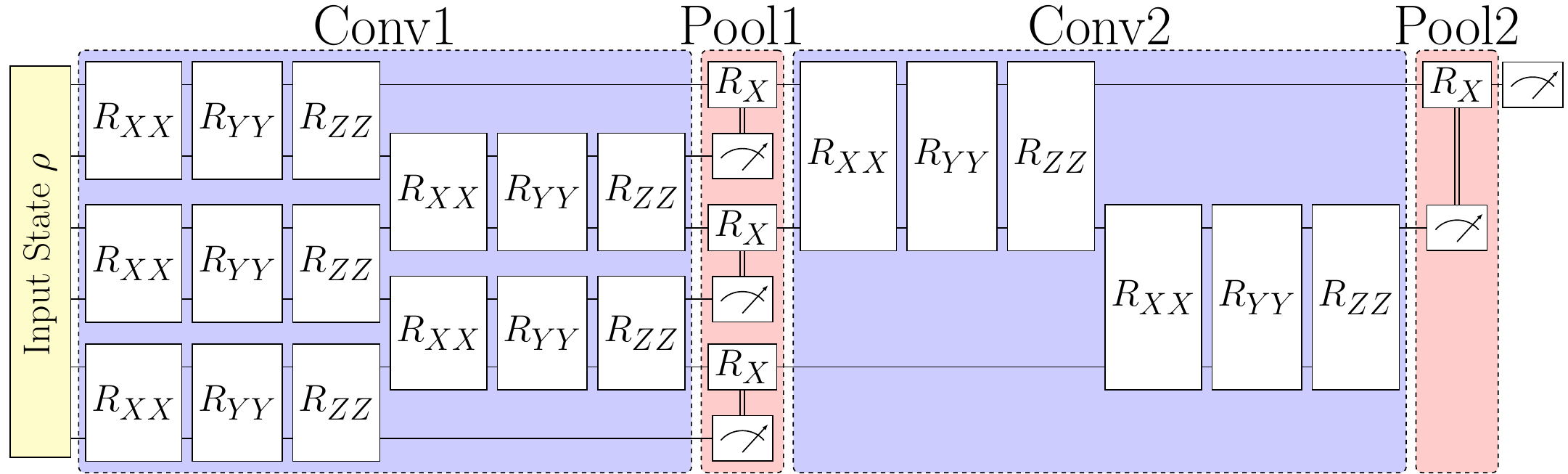}
	\caption{Architecture of  QCNN for the $6$-qubit ANNNI quantum phase classification: ground state preparation (yellow),  then two blocks of  convolutional (Conv) and pooling (Pool) layers and a measurement on the first qubit at the end. The pooling layers not only reduce the  size of the subsequent system, but also provide non-linearity for the whole circuits.}
\label{fig:6_qubit_QCNN_binary}
\end{figure}

The QCNN for the $6$-qubit ANNNI quantum phase classification is illustrated in Fig.~\ref{fig:6_qubit_QCNN_binary}. The yellow shaded represents the ground state preparation. The convolutional layers involve variational gates acting on pairs of qubits. In the pooling layers, one of the paired qubits is measured, and conditioned on the measurement result, a rotation is applied to the remaining qubit.

\begin{figure}[h]
	\centering
	\includegraphics[width=0.8\linewidth]{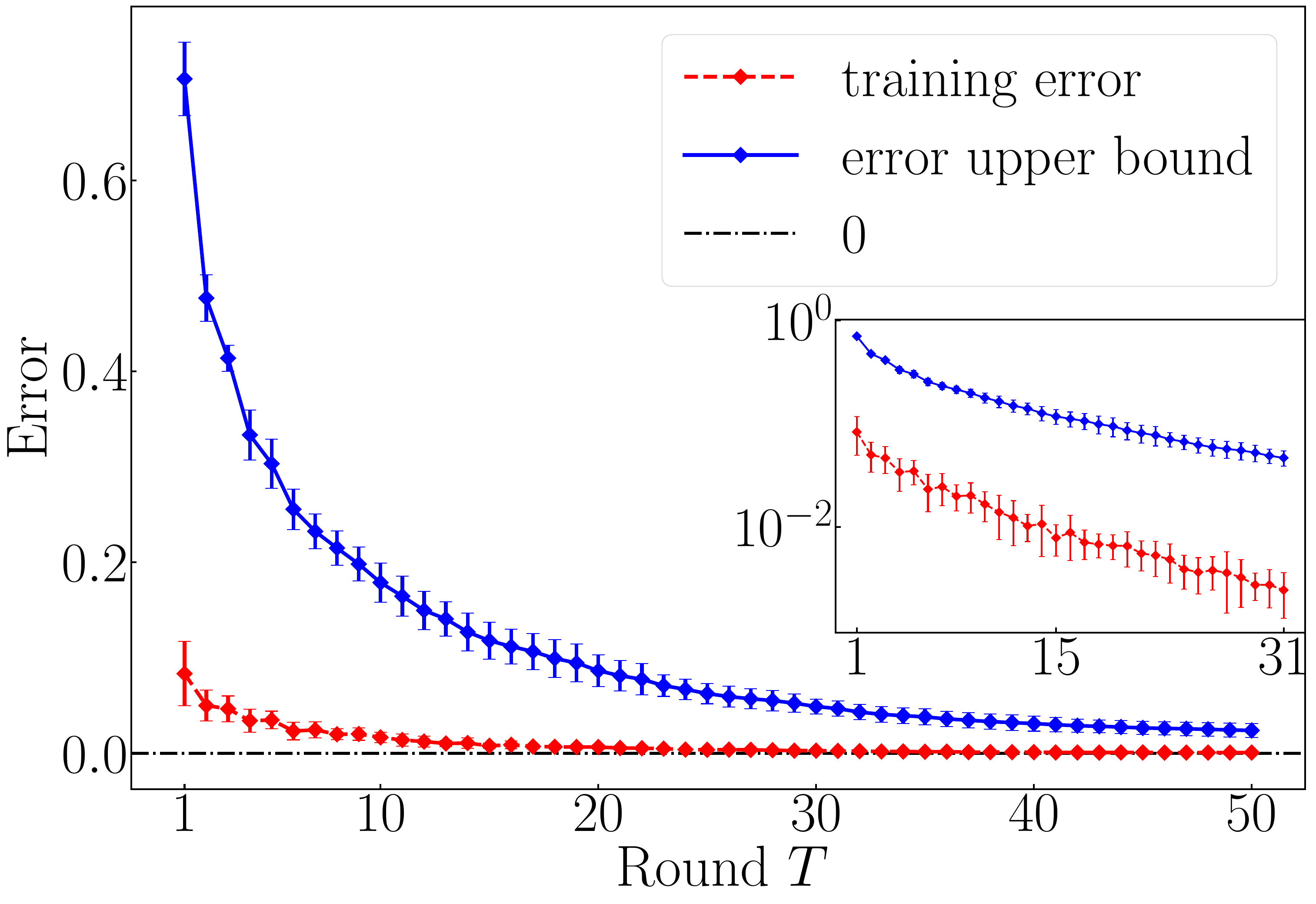}
	\caption{Training error and error upper bound versus the boosting round $T$. The blue line and red dashed illustrate the average of the error upper bound and the training error over 10 runs, respectively. The error bars show the corresponding standard deviations. Here, the inset depicts the exponential decrease of the training error in the first 31 rounds of boosting on a semi-log plot.}
\label{fig:train_bound_mean_std}
\end{figure}

We first verify that the first term $e^{{-2\sum_{t=1}^{T}(\frac{1}{2}-\epsilon_t)^2}}$ in the RHS of Eq.~(\ref{th01b})  gives an upper bound for the training error  $\widehat{R}_S\left( h\right)$. Consider a training data set $S={\left\lbrace \left( \rho\left(\kappa_i, h_i\right) , y_i\right) \right\rbrace} ^n_{i=1}$ consisting of $n=2000$ sample points. Set the maximum number of boosting rounds $T=50$, and repeat the demonstration for 10 times. The numerical results are demonstrated in Fig.~\ref{fig:train_bound_mean_std}. The blue line and red dashed depict the average of $e^{{-2\sum_{t=1}^{T}(\frac{1}{2}-\epsilon_t)^2}}$ and the training error $\widehat{R}_S\left( h\right)$ over 10 realizations, respectively. The error bars illustrate the corresponding standard deviations. We can see that as  $T$ grows, the difference between the error upper bound and the training error decreases.

\begin{figure}[h]
	\centering
	\includegraphics[width=0.8\linewidth]{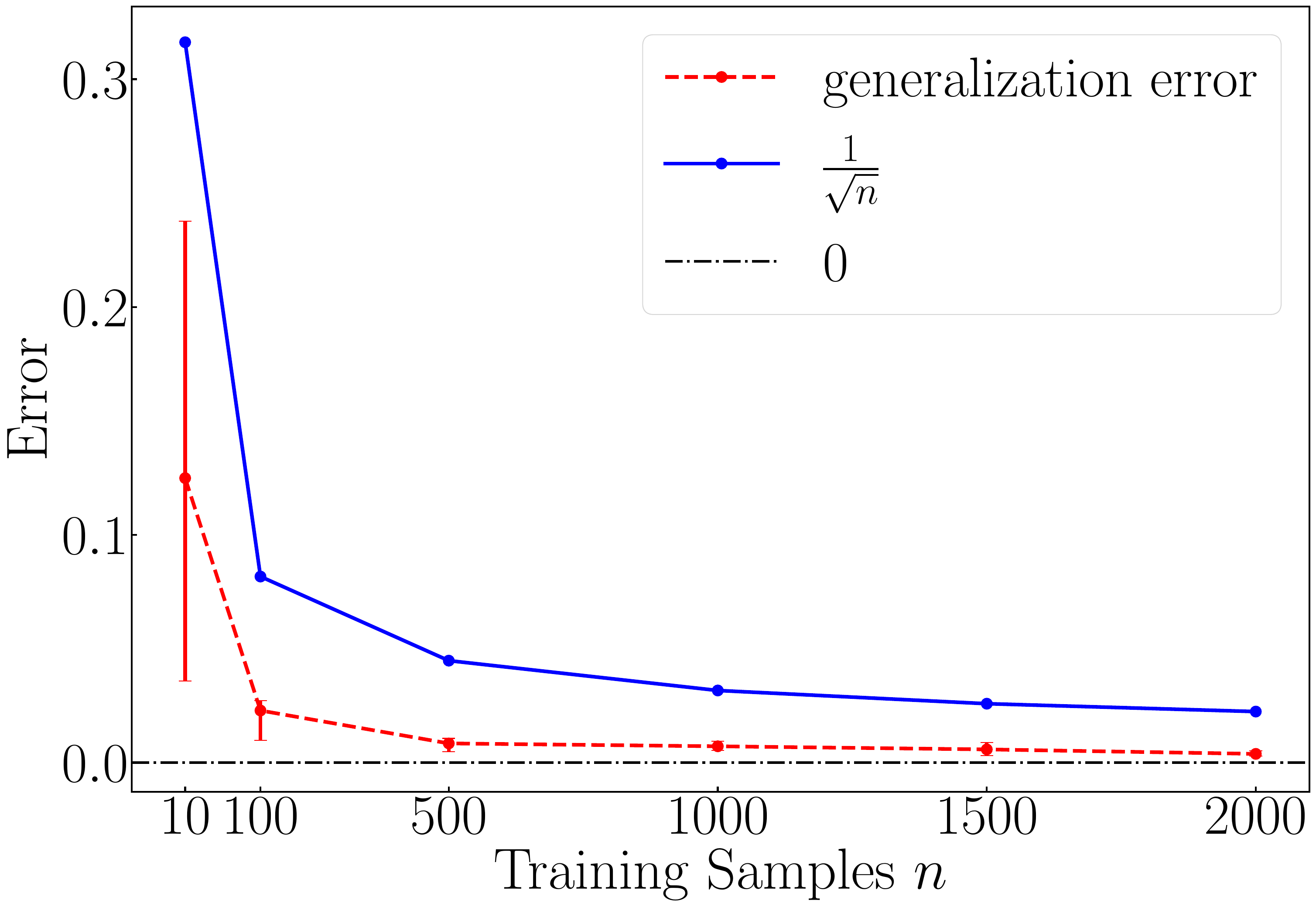}
	\caption{Generalization error versus different sizes of training data set. The red dashed line shows the average generalization error across 10  runs, and the error bars illustrate the ranges of the errors in the 10 realizations. The blue baseline represents $\frac{1}{\sqrt{n}}$, which corresponds to the asymptotic scaling of the generalization bound.}
\label{fig:gen_bound}
\end{figure}

We now validate that the second and third terms in the RHS of Eq.~(\ref{th01b})  give an upper bound of  the generalization error $\texttt{gen}\left( h\right)$. We consider different sizes $n$ of the training data set $S={\left\lbrace \left( \rho\left(\kappa_i, h_i\right) , y_i\right) \right\rbrace} ^n_{i=1}$. For each fixed size $n$, the training stops if a zero training error is obtained, which occurs when $n$ is small; otherwise, the maximum number of boosting rounds $T$ is set to be $T=50$. The size of the test data is 5000, and all the test samples are different from the training samples. The  demonstration is conducted 10 times and we present the numerical results in Fig.~\ref{fig:gen_bound}. The red dashed line depicts the average generalization error across the 10 runs, with the error bars indicating the ranges of errors in the 10 realizations. Moreover, the average generalization error is upper bounded by $\frac{1}{\sqrt{n}}$, aligning with the asymptotic scaling of the theoretical bound.

We have verified the learning risk bound for quantum AdaBoost on binary classification given by Theorem~\ref{th01}. We would like to point out that the risk bound should be further sharpened in the future, particularly, when $n/K$ is relatively small.


\section{\label{extension} Advantages of Quantum AdaBoost on 4-Class Classification}

\begin{figure*}[t]
	\flushleft
	\includegraphics[width=\textwidth]{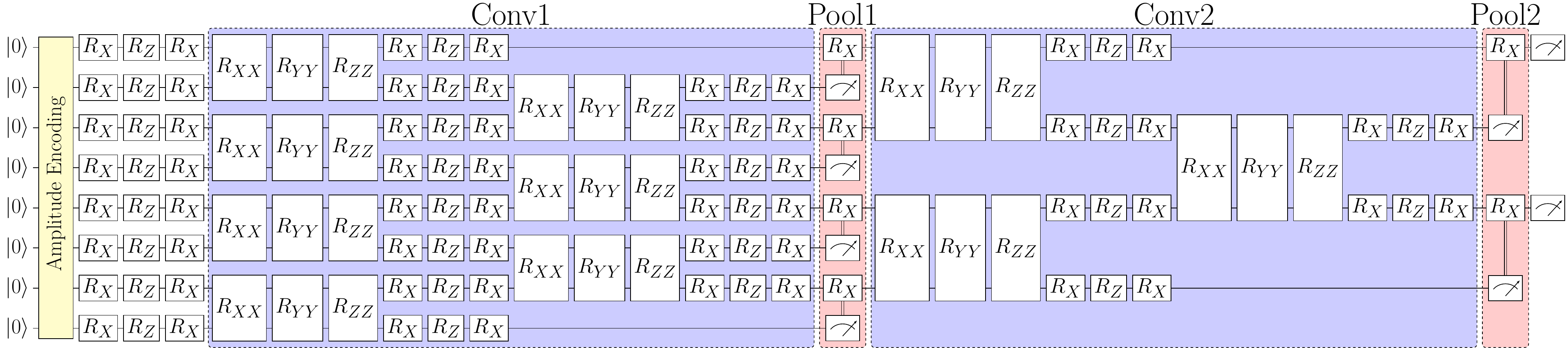}
	\caption{Architecture of QCNN for 4-class classification: amplitude encoding, then a set of single-qubit rotational gates followed by two blocks of convolutional (Conv) and pooling (Pool) layers, and a measurement on two qubits at the end.}
	\label{fig:QCNN}
\end{figure*}

In this section, we numerically demonstrate the advantages of quantum AdaBoost by  classifying the handwritten digits $\left\lbrace 0,1,2,3\right\rbrace $ in MNIST datasets \cite{lecun1998gradient}. The numerical simulations are conducted using the Python libraries PennyLane \cite{pennylane2018} and JAX \cite{jax2018github}. Specifically, we utilize
the `default.qubit' device in PennyLane to simulate noiseless circuits and the `default.mixed' device to simulate noisy circuits.

\subsection{Advantage of AdaBoost}\label{IVA}
In this subsection, we demonstrate the advantage of AdaBoost by comparing it with an unboosted primitive classifier.
\begin{figure}[b]
	\centering
	\includegraphics[width=0.8\linewidth]{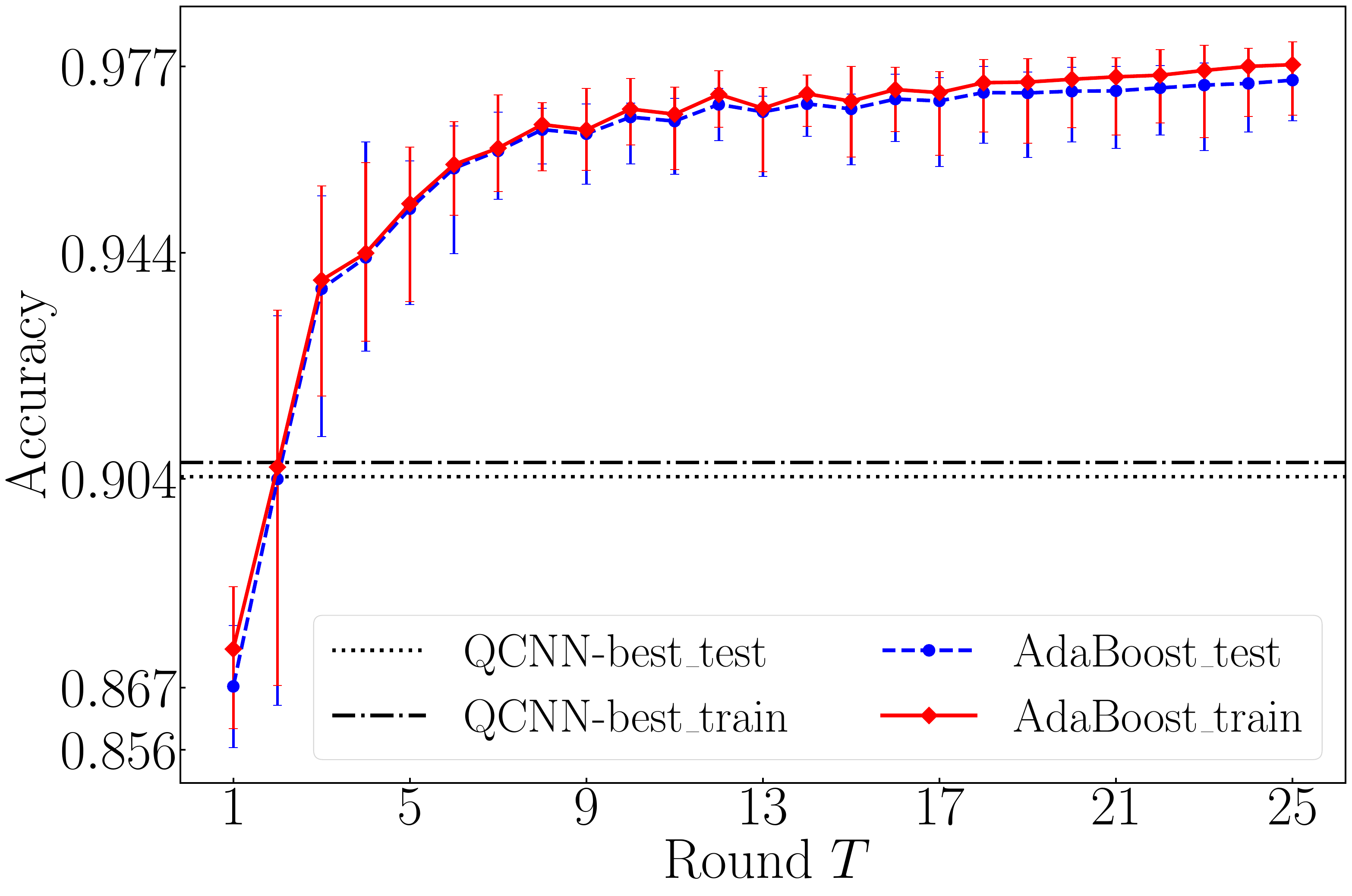}
	\caption{Accuracy  of  quantum AdaBoost and QCNN-best for 4-class classification. The blue dashed (red) line depicts the average test (training)  accuracy over 10 realizations of quantum AdaBoost.  The bars illustrate the ranges of the accuracies in the 10 runs. The black dotted (dash-dotted) depicts the test (training) accuracy of QCNN-best. It is clear that quantum AdaBoost outperforms QCNN-best after only 3 rounds of boosting.}
	\label{fig:noiseless}
\end{figure}

We randomly draw  $8000$ sample points for training and test, respectively, and the test samples are different from the training ones.
For each sampled image in MNIST, we first downsample it from $28 \times 28$ to $16 \times 16$, and then embed it into a quantum state of an $8$-qubit system using amplitude encoding. There are many other strategies for encoding classical data, e.g., angel encoding. It was proved that angle encoding will lead to the average encoded state concentrating on the maximally mixed state at an exponential speed on the depth of the encoding circuit \cite{Li2022Concentration}. Thus, to avoid the reduction of the encoded quantum state, we employ the amplitude encoding strategy. We would like to point out that how to realize the amplitude encoding is another question worth studying \cite{Nakaji2022,Mitsuda2024}. The QCNN for this 4-class classification is illustrated in Fig.~\ref{fig:QCNN}, which contains a total number of 120 variational parameters.

The demonstration is performed for 10 times. For quantum AdaBoost, the maximum number of boosting rounds $T$ is set to be 25.
To demonstrate the advantage of AdaBoost, we also consider a primitive classifier without any boosting which is referred to as QCNN-best. We train the QCNN-best by utilizing the cross-entropy loss function Eq.~(\ref{loss0}). It is optimized for $3000$ iterations. This is the same number as that in quantum AdaBoost, which has $T=25$ boosting rounds with each round having 120 iterations. The term {\it best} in QCNN-best means that we return the classifier having the best test performance among a total number of  $30000~ (10\times 3000)$ classifiers.

We demonstrate the  accuracy (one minus error) of quantum AdaBoost and QCNN-best  in Fig.~\ref{fig:noiseless}. The blue dashed (red) line illustrates the average test (training)  accuracy over 10 realizations of quantum AdaBoost, while the bars show the ranges of the accuracies in the 10  runs. The black dotted (dash-dotted) depicts the test (training) accuracy of QCNN-best. It is clear that  quantum AdaBoost outperforms QCNN-best  after only 3 rounds of boosting, and its test performance can always exceed $0.953$ (among 10 runs) after $T=7$ rounds of boosting. Moreover, the differences between the test and training accuracies of quantum AdaBoost are always below $0.01175$. This clearly demonstrates the good generalization property of  quantum AdaBoost. Thus, to improve the prediction accuracy, boosting is much better than simply increasing the number of training iterations.
	
 \subsection{Advantage over Other Ensemble Methods}
\begin{table*}[ht]
	\centering
	\caption{Comparison of  four different ensemble methods. The first row lists the training accuracy (acc.), and the second row lists the test accuracy. The values in the table represent the average values $\pm$ standard deviations.}
	\hspace{13cm}
	
	\begin{tabular}{c|c|c|c|c}
		\hline
		&QCNN+AdaBoost                             & QCNN+Bagging    & CNN+AdaBoost & CNN+Bagging \\
		\hline
		Training Acc. &0.977$\pm$0.004             & 0.891$\pm$0.005  & 0.984$\pm$0.009 & 0.981$\pm$0.001 \\
		\hline
		Test Acc.& \textbf{0.975$\pm$0.003}  &  0.884$\pm$0.005 & 0.965$\pm$0.007 & 0.970$\pm$0.001 \\
		\hline
	\end{tabular}
	\label{tab:comparison}
\end{table*}

In this subsection, we compare quantum AdaBoost (QCNN+AdaBoost) with three other ensemble methods. The first one is
quantum Bagging (QCNN+Bagging), the second is classical neural networks (CNN) with AdaBoost, referred to as CNN+AdaBoost, and the third one is CNN powered by Bagging, abbreviated as CNN+Bagging. The training data set, test data set, and the QCNN are the same as those in Section~\ref{IVA}.

The Bagging is one of the most popular ensemble approaches in classification problems~\cite{breiman1996bagging}. The prediction of Bagging is given through majority vote among the base classifiers. For quantum Bagging, we independently train $T$ different quantum base classifiers by utilizing the cross-entropy loss function Eq.~(\ref{loss0}) and the QCNN illustrated in Fig.~\ref{fig:QCNN}.  The difference between these quantum base classifiers lies in their training data sets, each of which is  obtained by resampling the original training data set for $8000$ times. Besides,  the initialization of QCNN parameters is different for different quantum base classifiers. Here, the number of base classifiers is set to be $T=25$, which is the number of boosting rounds of AdaBoost in Section~\ref{IVA}.

The CNN takes the form of $$f(x) = \sigma\Big{(}W_{2} \sigma(W_{1} x + b_1) + b_2\Big{)},$$ where $\sigma(\cdot)$ denotes the softmax function and $W_1 \in \mathbb{R}^{3 \times 256}, W_2 \in \mathbb{R}^{4 \times 3}, b_1 \in \mathbb{R}^{3}, b_2 \in \mathbb{R}^{4}$.
Note that CNN has 787 parameters, more than 6 times as many as QCNN (120 parameters).

For the four ensemble methods to be compared, we utilize the same numerical settings. Specifically, the optimizer is Adam, the learning rate is set to be 0.05, and all the parameters are initialized according to the standard normal distribution. We adopt the classifier having the smallest training error among the 120  iterations as the base classifier.  We perform each ensemble method for 10 times.

The numerical results are shown in Table~\ref{tab:comparison}. The values in the table represent the average values  and their standard deviations. We find that  both  CNN+AdaBoost and  CNN+Bagging are superior to QCNN+AdaBoost in terms of training accuracy. This is mainly because CNN has more variational parameters than  QCNN, so the expressibility of CNN is better than QCNN.
However, owing to the quantum advantage in generalization, quantum AdaBoost (QCNN+AdaBoost) has the best  test accuracy among the four ensemble methods.


 \subsection{Advantage in Error Mitigation}

In this subsection, we show that quantum AdaBoost can effectively mitigate the effect of noises.

\begin{figure}[h]
	\centering
	\includegraphics[width=\linewidth]{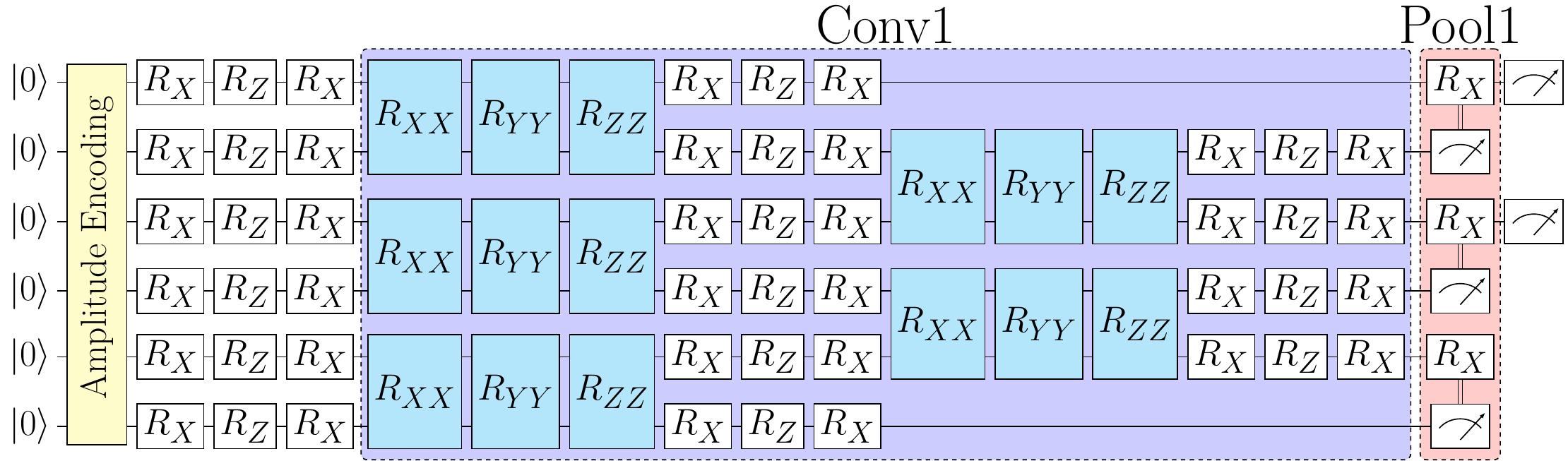}
	\caption{Architecture of  noisy QCNN for 4-class classification: amplitude encoding (yellow), then a set of single-qubit rotational gates followed by a convolutional (Conv) layer and a pooling (Pool) layer, and a measurement on two qubits at the end. The two-qubit gates (blue) are noisy.}
	\label{fig:6_qubit_QCNN1}
\end{figure}

We also consider the classification of the handwritten digits $\left\lbrace 0,1,2,3\right\rbrace $ in MNIST datasets. The training and test data sets are randomly drawn from the MNIST, and the test samples are different from the training samples. Both the training and test data sets are 1000 in size. For each sampled image, we first downsample it from $28 \times 28$ to $8 \times 8$ and then use the amplitude encoding to embed it into a quantum state of a $6$-qubit system. The  QCNN is illustrated in Fig.~\ref{fig:6_qubit_QCNN1}. In practice, implementing single-qubit gates can achieve a high level of fidelity, while the fidelity of implementing two-qubit gates is still relatively lower. Thus, we assume that the two-qubit gates (colored blue)  in Fig.~\ref{fig:6_qubit_QCNN1} are noisy.

To simulate the effect of noises in QCNN,  we consider three typical types of quantum noise channels:
\begin{itemize}
  \item Depolarizing channel: $ \mathcal{E}_1(\rho)=\frac{p}{2}I+(1-p) \rho$, with $p=0.03$;
  \item Amplitude damping: $\mathcal{E}_{2}(\rho)=E_{0}^A \rho E_{0}^{A\dagger}+E_{1}^A \rho E_{1}^{A\dagger},$ where $\begin{array}{l}
  E_{0}^A=\left[\begin{array}{cc}
  1 & 0 \\
  0 & \sqrt{1-\gamma}
  \end{array}\right] ,
  E_{1}^A=\left[\begin{array}{cc}
  0 & \sqrt{\gamma} \\
  0 & 0
  \end{array}\right],
  \end{array}$ and $\gamma=0.03$;
  \item Phase damping: $  \mathcal{E}_{3}(\rho)=E_{0}^P \rho E_{0}^{P\dagger}+E_{1}^P \rho E_{1}^{P\dagger},$ where $\begin{array}{l}
  E_{0}^P=\left[\begin{array}{cc}
  1 & 0 \\
  0 & \sqrt{1-\lambda}
  \end{array}\right],
  E_{1}^P=\left[\begin{array}{cc}
  0 & 0 \\
  0 & \sqrt{\lambda}
  \end{array}\right]
  \end{array}$, and $\lambda=0.03$.
\end{itemize}
A noisy  two-qubit gate  in Fig.~\ref{fig:6_qubit_QCNN1} represents that after the ideal gate we apply a noise channel on each of the two qubits.

\begin{figure*}[htpb]
	\centering
	\includegraphics[width=0.8\textwidth]{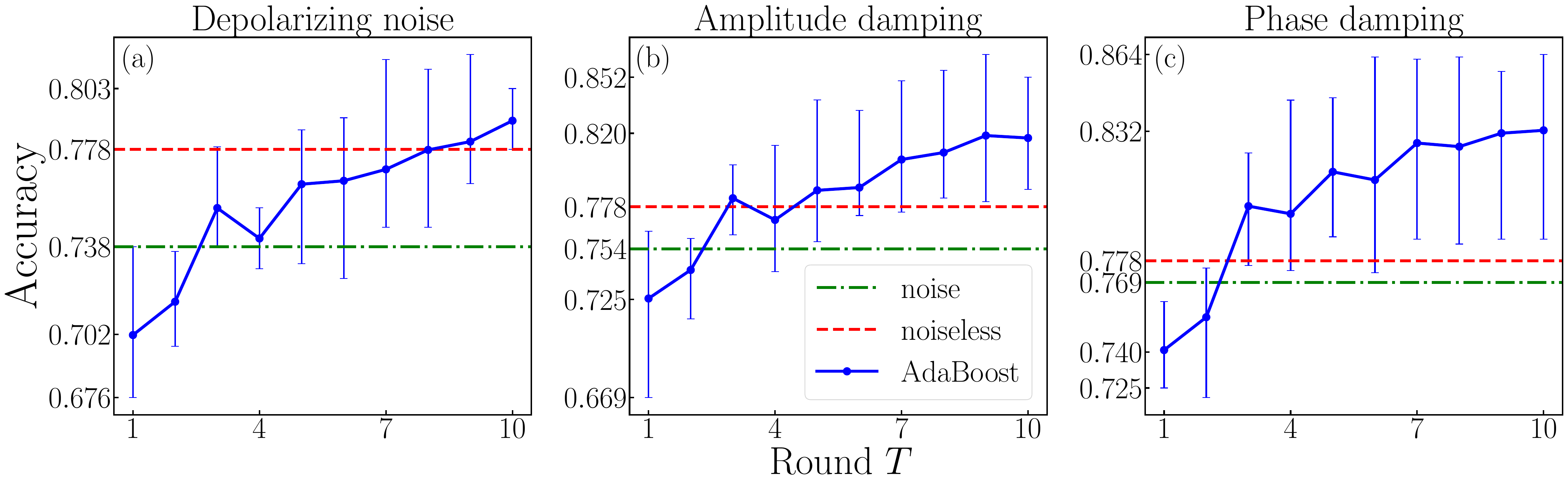}
	\caption{Comparison of test accuracy. (a) QCNN with depolarizing. (b) QCNN with amplitude damping. (c) QCNN with phase damping.  The blue line depicts the average test accuracy of  quantum AdaBoost with noisy QCNN. The bars show the ranges of the test accuracies in 10 runs. The green dash-dotted (red dashed) line describes the best test accuracy of the primitive method with noisy (noiseless) QCNN. Quantum AdaBoost outperforms the noiseless primitive classifier  after only a few rounds of boosting.}
	\label{fig:noise}
\end{figure*}

We compare the test accuracy of  quantum AdaBoost, the primitive classifier with noisy QCNN (Fig.~\ref{fig:6_qubit_QCNN1}), and the primitive classifier with noiseless QCNN where all gates in Fig.~\ref{fig:6_qubit_QCNN1} are ideal. We perform 10 runs  for each method. In quantum AdaBoost, the maximum number of boosting rounds $T$ is set to be 10, with each round having 120 iterations.  For the primitive methods, we conduct 1200 iterations, which is as many as the total iterations in quantum AdaBoost. The other  numerical settings are the same as before.

We illustrate the test accuracy of the three methods in Fig.~\ref{fig:noise}. For the two primitive methods, we return the classifier having the best test accuracy among the 12000 classifiers. The green dash-dotted (red dashed) line describes the best test accuracy of the primitive classifier with noisy (noiseless) QCNN. The blue line depicts the average test accuracy of  quantum AdaBoost with noisy QCNN. The  bars demonstrate the ranges of the test accuracies in 10 runs.  We can see that for the three typical types of noise, quantum AdaBoost outperforms the noiseless unboosted primitive classifier after only a few rounds of boosting. The reason is twofold. First, since the QCNN is shallow, the degradation of quantum classifiers caused by noises are relatively small. This is evidenced by the small differences in test accuracy between the noisy and noiseless primitive classifiers in Fig.~\ref{fig:noise}. Second, although noises do degrade weak classifiers, as long as they still outperform the random guess, according to Theorem~\ref{th01}, the weak classifiers can be boosted to obtain a strong classifier whose performance improves as the number of boosting rounds and the sample size increase.  Thus, we can see that quantum AdaBoost can help mitigate the impact of noises in PQCs. This good property of quantum AdaBoost is particularly useful in the NISQ era.

\section{\label{conclusion}Conclusion}
Quantum machine learning has good generalization property while its trainability is generally poor. Ensemble methods are particularly appropriate to enhance the trainability of quantum machine learning, and in turn improve the prediction accuracy. In this paper we theoretically establish and numerically verify a prediction guarantee for binary quantum AdaBoost. We also demonstrate that for a 4-class classification problem, quantum AdaBoost not only can achieve a higher level of prediction accuracy compared to the unboosted primitive classifier and three other ensemble approaches, but also help mitigate the impact of noises.  For future work, it is interesting to incorporate ensemble methods to solve other practical tasks.

\section*{Acknowledgments}
B.~Q. acknowledges the support of the National Natural Science Foundation of China (No.~61773370). D.~D. acknowledges the support of the Australian Research Council’s Future Fellowship funding scheme under Project No.~FT220100656.

\appendix
\section*{Appendices}
In the Appendices, we present the lemmas supporting our main theorem in Appendix~\ref{lemmas}, and provide the formal proof of Theorem~\ref{th01} and Lemma~\ref{radh} in Appendix~\ref{proofth01} and \ref{proofthradh}, respectively.

\section{\label{lemmas}Technical Lemmas}

We first introduce a lemma that can be used to bound the training error of binary quantum AdaBoost.
\begin{lemma}\label{th72}
	{\em (Theorem 7.2,  \cite{mohri2018foundations})}
	The training error of the binary classifier returned by AdaBoost verifies:
	\begin{equation}\label{key}
	\widehat{R}_{S}\left( h\right)\le e^{-2\sum_{t=1}^{T}{{\left( \frac{1}{2}-\epsilon_{t}\right)} }^2}.
	\end{equation}
\end{lemma}

Then we recall a well-known result in machine learning which provides an upper bound on the generalization error.

\begin{lemma}\label{th33}
	{\em (Theorem 3.3,  \cite{mohri2018foundations})}	Let $\mathcal{G}$ be a family of functions mapping from $\mathcal{Z}$ to $\left[ 0,1\right] $. Then, for any $\delta\, \textgreater \,0$, with probability at least $1-\delta$ over the draw of an i.i.d. sample set  $S=\left\lbrace z_{i}\right\rbrace _{i=1}^{n}$ according to an unknown distribution $\mathcal{D}$,  the following inequality holds for all $g\in \mathcal{G}$:
	\begin{equation*}\label{key}
		\mathop{\mathbb{E}}\limits_{z\sim \mathcal{D}}{\left[g\left( z\right) \right] }\le\frac{1}{n}\sum_{i=1}^{n}{g\left( z_{i}\right)}+2\mathfrak{R}_{n}\left( \mathcal{G}\right)+\sqrt{\frac{\log{\frac{1}{\delta}} }{2n}},
	\end{equation*}
 \onecolumngrid

	\noindent where $\mathfrak{R}_{n}\left( \mathcal{G}\right)$ denotes the expectation of the empirical Rademacher complexity $\widehat{\mathfrak{R}}_{S}\left( \mathcal{G}\right)$ of $\mathcal{G}$, defined as
	\begin{equation}\label{key}
	\mathfrak{R}_{n}\left( \mathcal{G}\right)=\mathop{\mathbb{E}}\limits_{S\sim {\mathcal{D}}^{n}}\left[ \widehat{\mathfrak{R}}_{S}\left( \mathcal{G}\right) \right] =\mathop{\mathbb{E}}\limits_{S\sim {\mathcal{D}}^{n}}{\mathop{\mathbb{E}}\limits_{\boldsymbol{\sigma}}{\left[ \sup_{g\in\mathcal{G} }\frac{1}{n}\sum_{i=1}^{n}{\sigma_{i}g\left(z_{i} \right)  }\right] ,}}
	\end{equation}
	where $\boldsymbol{\sigma}=\left( \sigma_{1},\cdots,\sigma_{n}\right)^{\top} $, with $\sigma_i$s  independent uniform random variables taking values in $\left\lbrace -1,+1\right\rbrace $.
\end{lemma}

The following lemma relates
the empirical Rademacher complexity of a new set of composite functions of a hypothesis in $\mathcal{H}$ and a Lipschitz function
to the empirical Rademacher complexity of the hypothesis set $\mathcal{H}$.

\begin{lemma}\label{talagrand}
	\em {\em(Lemma 5.7,  \cite{mohri2018foundations})}	Let $\Phi_{1},\cdots,\Phi_{n}$ be $l$-Lipschitz functions from $\mathbb{R}$ to $\mathbb{R}$ and $\boldsymbol{\sigma}=\left( \sigma_{1},\cdots,\sigma_{n}\right)^{\top} $ with $\sigma_i$s  independent uniform random variables taking values in $\left\lbrace -1,+1\right\rbrace $. Then, for any hypothesis set $\mathcal{H}$ of real-valued functions, the following inequality holds:
	\begin{equation*}\label{key}
		\frac{1}{n}\mathop{\mathbb{E}}\limits_{\boldsymbol{\sigma}}{\left[ \sup_{h\in\mathcal{H} }\sum_{i=1}^{n}{\sigma_{i}\left( \Phi_{i}\circ h \right)\left(x_{i} \right)  }\right] }\le \frac{l}{n}\mathop{\mathbb{E}}\limits_{\boldsymbol{\sigma}}{\left[ \sup_{h\in\mathcal{H} }\sum_{i=1}^{n}{\sigma_{i}h\left(x_{i} \right)  }\right] }.
	\end{equation*}
\end{lemma}

In our work, the hypothesis set $\mathcal{H}=\left\lbrace h_{\boldsymbol{\theta}}:\boldsymbol{\theta}\in {\left[ 0,2\pi\right]}^K  \right\rbrace $, which is composed of PQC-based hypothesis $h_{\boldsymbol{\theta}}$ defined  in the form of Eq.~(\ref{hypothesis2}). We provide an upper bound of its Rademacher complexity in the following lemma (see Appendix~\ref{proofthradh} for a detailed proof).

\begin{lemma}\label{radh}
	For the quantum hypothesis set $\mathcal{H}=\left\lbrace h_{\boldsymbol{\theta}}:\boldsymbol{\theta}\in {\left[ 0,2\pi\right]}^K  \right\rbrace $ with $h_{\boldsymbol{\theta}}$ being defined as in Eq.~(\ref{hypothesis2}), its Rademacher complexity  can be upper bounded by a  function of the number of independent trainable quantum gates $K$ and the sample size $n$:
	\begin{equation}\label{radhb}
	\mathfrak{R}_{n}\left( \mathcal{H}\right)\le 6\sqrt{\frac{K\log 7K}{n}}+2\sqrt{\frac{K}{n}}.
	\end{equation}
\end{lemma}

\section{\label{proofth01}Proof of Theorem~\ref{th01}}
\begin{proof}
	In the context of binary classification, it is evident that evaluating the prediction error of the binary quantum AdaBoost classifier returned by Algorithm~\ref{algorithm2} is equivalent to evaluating $R\left( h\right)$ where $h=\sum_{t=1}^{T}{\alpha_{t}h_{\boldsymbol{\theta}_t}}$, after removing the outermost sign function from the predictor. However, $h$
	does not belong to the quantum hypothesis set  $\mathcal{H}=\left\lbrace h_{\boldsymbol{\theta}}:\boldsymbol{\theta}\in {\left[ 0,2\pi\right]}^K  \right\rbrace $ or its convex hull denoted by
	\begin{equation}\label{key}
	{\rm conv}\left( \mathcal{H}\right)\triangleq\left\lbrace \sum_{k=1}^{p}{\mu_k h_k}: p\ge 1, \mu_k \ge 0, h_k \in \mathcal{H}, \sum_{k=1}^{p}{\mu_k} \le 1\right\rbrace.
	\end{equation}
	To address this, we consider the normalized version of $h$ denoted by
	\begin{equation}\label{key}
	\bar{h}\triangleq\frac{h}{\sum_{t=1}^{T}{\alpha_{t}}}\in{\rm conv}\left( \mathcal{H}\right).
	\end{equation}
	Given Eqs.~(\ref{prediction}) and (\ref{equality}) and the fact that $\text{sgn}\left( h\right)=\text{sgn}\left(\bar{h}\right)$, we have
	\begin{equation}\label{relations}
	R\left( h\right)=R\left(\bar{h}\right),~\text{and}~\widehat{R}_{S}\left(h\right)=\widehat{R}_{S}\left(\bar{h}\right).
	\end{equation}
	
	To bound $R\left(\bar{h}\right)$, let
	\begin{equation}\label{key}
	\mathcal{G}=\left\lbrace \mathbb{I}_{y\bar{h}\left( \boldsymbol{x}\right)\le 0 }:\left( \rho\left( \boldsymbol{x}\right) ,y\right) \in\mathcal{X}\times\mathcal{Y},\, \bar{h}=\frac{\sum_{t=1}^{T}{\alpha_{t}h_{\boldsymbol{\theta}_t}}}{\sum_{t=1}^{T}{\alpha_{t}}},\, h_{\boldsymbol{\theta}_t}\in \mathcal{H},\,t\in\left[T \right]  \right\rbrace.
	\end{equation}
	Then according to Lemma~\ref{th33}, it yields
	\begin{equation}\label{barf}
	R\left(\bar{h}\right) \le \widehat{R}_{S}\left(\bar{h}\right) +2\mathfrak{R}_{n}\left( \mathcal{G}\right)+\sqrt{\frac{\log{\frac{1}{\delta}} }{2n}}.
	\end{equation}
	Since the zero-one loss function is 1-Lipschitz, from Lemma~\ref{talagrand}, we have \begin{equation}\label{radg}
	\mathfrak{R}_{n}\left( \mathcal{G}\right) \leq \mathfrak{R}_{n}\left( {\rm conv}\left( \mathcal{H}\right)\right).
	\end{equation}
	Moreover, according to Lemma 7.4 in \cite{mohri2018foundations}, \begin{equation}\label{radconvh}
	\mathfrak{R}_{n}\left( {\rm conv}\left( \mathcal{H}\right)\right) = \mathfrak{R}_{n}\left( \mathcal{H}\right).
	\end{equation}
	
	Thus, by combining Eqs.~(\ref{relations}), (\ref{barf})-(\ref{radconvh}), a general upper bound can be derived as
	\begin{equation}\label{rf}
	R\left(h\right) \le \widehat{R}_{S}\left(h\right) +2\mathfrak{R}_{n}\left( \mathcal{H}\right)+\sqrt{\frac{\log{\frac{1}{\delta}} }{2n}}.
	\end{equation}
	Furthermore, by incorporating Lemma~\ref{th72} and Lemma~\ref{radh}, Theorem~\ref{th01} can be established.
\end{proof}

\section{\label{proofthradh}Proof of Lemma~\ref{radh}}

We first introduce the notion of covering number, which is a  complexity measure that has been widely used in machine learning.
\begin{definition}
	(Covering nets and covering numbers \cite{Dudley2014}) Let $\left( \mathcal{H}, d \right)$ be a metric space. Consider a subset $\mathcal{K}\subset\mathcal{H} $ and let $\epsilon\textgreater0$. A subset $\mathcal{N} \subseteq \mathcal{K} $ is called an $\epsilon$-net of $\mathcal{K}$ if every point in $\mathcal{K}$ is within a distance $\epsilon$ of some point of $\mathcal{N}$, i.e.,
	\begin{equation*}
		\forall~ x\in \mathcal{K},\, \exists~ y\in \mathcal{N}: ~d\left(x,y \right) \le \epsilon.
	\end{equation*}
	The smallest possible cardinality of an $\epsilon$-net of $\mathcal{K}$ is called the covering number of $\mathcal{K}$, and is denoted by $N\left( \mathcal{K},\epsilon,d\right)$.
\end{definition}

For example, for Euclidean space $\left( \mathbb{R}^{K}, \Vert \cdot\Vert _{\infty} \right)$, the covering number $N\left( \left[-\pi,\pi \right] ^{K},\epsilon,\Vert \cdot\Vert _{\infty}\right)$ is equal to ${\lceil \pi /\epsilon\rceil}^{K}$, where $\lceil \cdot \rceil$ denotes the rounding up function. Intuitively, the hypercube $\left[-\pi,\pi \right] ^{K}$ can be covered by a number of ${\lceil \pi /\epsilon \rceil}^{K}$ $K$-dimensional hypercubes whose sides have the same length $2\epsilon$.

Then we introduce several technical lemmas that will be used in the proof of Lemma~\ref{radh}.

The following lemma relates the distance between two unitary operators measured by the spectral norm to the distance between  their corresponding unitary channels  measured by the diamond norm.

\begin{lemma}\label{transnorm}
	{\rm  (Lemma 4, Supplementary Information for \cite{caro2022gen})} Let $\mathcal{U}\left(\rho\right)=U\rho U^{\dag}$ and $\mathcal{V}\left(\rho\right)=V\rho V^{\dag}$ be unitary channels. Then,
	\begin{equation*}
		\frac{1}{2}\Vert\mathcal{U}-\mathcal{V}\Vert_{\diamond}\le \Vert U-V\Vert.
	\end{equation*}
	Here, $\Vert\cdot\Vert$ denotes the spectral norm, and the diamond norm of a quantum unitary channel $\mathcal{U}$ is defined as $$\Vert\mathcal{U}\Vert_{\diamond}=\mathop{\max }\limits_{\left|\psi\rangle\langle\psi\right|}\Vert\mathcal{U}\left( \left|\psi\rangle\langle\psi\right|\right) \Vert_1,$$
	with $\Vert A \Vert_1={\rm Tr}\left[ \sqrt{A^\dag A}\right] $ being the trace norm.
\end{lemma}

The following lemma translates the distance between $J$-qubit rotational operators to the distance of their corresponding angles.

\begin{lemma}\label{transR}
	{\rm (Distance between rotational operators)}
	Given an arbitrary $J$-qubit Pauli tensor product $P\in{\left\lbrace I,X,Y,Z\right\rbrace}^{\otimes J} $ and two arbitrary angles $\theta,\tilde{\theta}\in \left[ 0,2\pi\right]$, the corresponding $J$-qubit rotational operators are $R\left(\theta\right)=e^{-i\frac{\theta}{2} P}$ and $R\left( \tilde{\theta}\right) =e^{-i\frac{\tilde{\theta}}{2} P}$, respectively. Then, the distance between the two operators measured by the  spectral norm can be upper bounded as
	\begin{equation*}
		\Big{\|}R\left(\theta\right)-R\left( \tilde{\theta}\right)\Big{\|}\le \frac{1}{2}\lvert\theta-\tilde{\theta}\rvert.
	\end{equation*}
\end{lemma}
\begin{proof}
	According to the definition of rotational operators, we have
	\begin{equation*}
		R\left(\theta\right)-R\left( \tilde{\theta}\right)=\left( \cos\frac{\theta}{2}-\cos\frac{\tilde{\theta}}{2}\right) I^{\otimes J}-i\left( \sin\frac{\theta}{2}-\sin\frac{\tilde{\theta}}{2}\right)P,
	\end{equation*}
	whose singular value is $2\left|\sin\frac{\theta-\tilde{\theta}}{4}\right|$ with $2^J$ multiplicity. Thus,
	\begin{equation*}
		\Big{\|} R\left(\theta\right)-R\left( \tilde{\theta}\right)\Big{\|}=2\left|\sin\frac{\theta-\tilde{\theta}}{4}\right|\le \frac{1}{2}\lvert\theta-\tilde{\theta}\rvert.
	\end{equation*}
\end{proof}

In the main text, we only consider the ideally unitary channel for simplicity, namely, the PQC-based hypothesis is in the form of $h_{\boldsymbol{\theta}}\left(\boldsymbol{x} \right) =\text{sgn}({\rm Tr}\left[ Z U\left(\boldsymbol{\theta} \right)\rho\left( \boldsymbol{x}\right)U^{\dagger}\left(\boldsymbol{\theta} \right)\right])$.   To describe the noise effect, a general quantum channel $\mathcal{A}$ is defined by a linear map $\mathcal{A}: \mathcal{L}\left( \mathcal{H}_A\right) \rightarrow \mathcal{L}\left( \mathcal{H}_B\right) $, which is completely positive and trace preserving (CPTP). For a quantum channel $\mathcal{A}$, the diamond norm is defined as
$$\Vert\mathcal{A}\Vert_{\diamond}=\mathop{\sup }\limits_{\rho\in\mathcal{D}\left( \mathcal{H}_{RA}\right) }\Vert\left( \mathcal{I}_R\otimes\mathcal{A}\right) \left( \rho\right) \Vert_1,$$
where $\mathcal{D}\left( \mathcal{H}_{RA}\right) $ denotes the set of  density operators acting on the Hilbert space $\mathcal{H}_{RA}=\mathcal{H}_R\otimes \mathcal{H}_A$, and $\mathcal{I}_R$ is the identity map on the reference  system $\mathcal{H}_R$, whose dimension can be arbitrary as long as the operator $\mathcal{I}_R\otimes\mathcal{A}$ is positive semi-definite.

The following lemma can help generalize the results of Lemma~\ref{radh} and Theorem~\ref{th01} from the case of  unitary quantum  channels described in the main text to those of noisy quantum channels. This means that the hypothesis functions can be generalized to
\begin{equation*}
	h_{\boldsymbol{\theta}}\left(\boldsymbol{x} \right) ={\rm sgn}({\rm Tr}\left[ Z \mathcal{E}_{\boldsymbol{\theta}}\left( \rho\left( \boldsymbol{x}\right)\right) \right]),
\end{equation*}
where the noisy channel $\mathcal{E}_{\boldsymbol{\theta}}$ is composed of $K$ trainable multi-qubit rotational gates and arbitrarily many non-trainable gates.
Notice that the general case can be reduced to the ideally unitary case by letting $\mathcal{E}_{\boldsymbol{\theta}} \left( \rho\right)=U\left( \boldsymbol{\theta}\right) \rho U^{\dag}\left( \boldsymbol{\theta}\right) $.


\begin{lemma}\label{subadd}
	{\rm (Subadditivity of diamond distance; Proposition 3.48, \cite{watrous2018theory})}
	For any quantum channels $\mathcal{A}$, $\mathcal{B}$, $\mathcal{C}$, $\mathcal{D}$, where $\mathcal{B}$ and $\mathcal{D}$ map from $n$-qubit to $m$-qubit systems and $\mathcal{A}$ and $\mathcal{C}$ map from $m$-qubit to $k$-qubit systems, we have
	\begin{equation*}
		\Vert \mathcal{A}\mathcal{B} - \mathcal{C}\mathcal{D}\Vert _{\diamond}\le \Vert \mathcal{A} - \mathcal{C}\Vert _{\diamond}+\Vert \mathcal{B} -\mathcal{D}\Vert _{\diamond}.
	\end{equation*}
\end{lemma}

The following lemma enables us to employ the covering number of one metric space to bound the covering number of another metric space.

\begin{lemma}\label{covernum}
	{\rm (Covering numbers of two metric spaces; Lemma 3, \cite{Barthel2018})} Let $\left( \mathcal{H}_{1}, d_{1} \right)$ and $\left( \mathcal{H}_{2}, d_{2} \right)$ be two metric spaces and $f:\mathcal{H}_{1}\rightarrow\mathcal{H}_{2}$ be bi-Lipschitz such that
	\begin{equation*}
		d_{2}\big{(}f\left( x\right), f\left( y\right) \big{)} \le K d_{1}\left(x,y \right), \forall ~x,y \in \mathcal{H}_{1},
	\end{equation*}
	where $K$ is a constant. Then their covering numbers obey the following inequality as
	\begin{equation*}
		N\left( \mathcal{H}_{2},\epsilon,d_{2}\right) \le N\left( \mathcal{H}_{1},\epsilon/K,d_{1}\right).
	\end{equation*}
	
\end{lemma}

According to the above lemmas, we can derive the covering number of general noisy quantum models.
\begin{lemma}{\label{TH1}}
	{\em (Covering number of  noisy quantum models)} If each element of $\boldsymbol{\theta}$ is selected from $\left[ -\pi, \pi\right] $, then the covering number of the set of  quantum channels $\left\lbrace \mathcal{E}_{\boldsymbol{\theta}}\right\rbrace $, each of which is composed of $K$ trainable multi-qubit rotational gates and arbitrarily many non-trainable quantum channels, can be upper bounded as
	\begin{equation*}
		N\left(\mathcal{E}_{\boldsymbol{\theta}},\epsilon,\Vert \cdot\Vert _{\diamond}\right)\le \bigg{\lceil}\frac{\pi K}{\epsilon}\bigg{\rceil}^{K}.
	\end{equation*}
\end{lemma}
\begin{proof}
	From the structure of $\mathcal{E}_{\boldsymbol{\theta}}$, we have
	\begin{align}
		\Vert \mathcal{E}_{\boldsymbol{\theta}}-\mathcal{E}_{\tilde{\boldsymbol{\theta}}}\Vert _{\diamond}&\le \sum_{k=1}^{K}{\Big{\|}  \mathcal{R}_{k}\left( \theta_{k}\right) -\mathcal{R}_{k}\left(  \tilde{\theta}_{k}\right)  \Big{\|}  _{\diamond}}\label{separate1}
		\\
		&\le 2 \sum_{k=1}^{K}{\Big{\|} R_{k}^{\left(i_{k}, j_{k}\right) }\left(\theta_{k} \right)-R_{k}^{\left(i_{k}, j_{k}\right) }\left( \tilde{\theta}_{k} \right) }\Big{\|} \label{trans1}
		\\
		&\le \sum_{k=1}^{K}{\lvert\theta_{k}-\tilde{\theta}_{k}\rvert}\label{trans2}\\ &=\Vert \boldsymbol{\theta}-\tilde{\boldsymbol{\theta}}\Vert _{1}
		\\
		&\le K\Vert \boldsymbol{\theta}-\tilde{\boldsymbol{\theta}}\Vert_{\infty},\label{trans3}
	\end{align}
	where $\mathcal{R}_{k}\left( \theta_{k}\right) $ denotes the quantum channel corresponding to the $J$-qubit rotational operator $R_{k}^{\left(i_{k}, j_{k}\right) }\left(\theta_{k} \right)$. Here, Eq.~(\ref{separate1}) is derived by repeatedly using Lemma~\ref{subadd} to erase the non-trainable quantum channels,  Eqs.~(\ref{trans1}) and (\ref{trans2}) are obtained from Lemma~\ref{transnorm} and Lemma~\ref{transR}, respectively, and  Eq.~(\ref{trans3}) owes to the relation  between the $l_{1}$-norm and $l_{\infty}$-norm for vectors in $\mathbb{R}^{K}$.
	
	Thus, according to Lemma~\ref{covernum}, we have
	\begin{equation*}
		N\left(\mathcal{E}_{\boldsymbol{\theta}},\epsilon,\Vert \cdot\Vert _{\diamond}\right)\le N\left(\boldsymbol{\theta},\frac{\epsilon}{K},\Vert \cdot\Vert _{\infty}\right)=\bigg{\lceil}\frac{\pi K}{\epsilon}\bigg{\rceil}^{K}.
	\end{equation*}
\end{proof}

Now we prove Lemma~\ref{radh} by leveraging the technique of proving Theorem 6 in Supplementary Information of Ref.~\cite{caro2022gen}.

\begin{proof}
	To bound the  Rademacher complexity
	\begin{equation*}
		\mathfrak{R}_{n}\left( \mathcal{H}\right)=
		\mathop{\mathbb{E}}\limits_{S\sim {\mathcal{D}}^{n}}{\mathop{\mathbb{E}}\limits_{\boldsymbol{\sigma}}{\left[ \sup_{h_{\boldsymbol{\theta}}\in\mathcal{H} }\frac{1}{n}\sum_{i=1}^{n}{\sigma_{i}h_{\boldsymbol{\theta}}\left(\boldsymbol{x}_{i} \right)  }\right] }},
	\end{equation*}
	the main idea is to use the chaining technique to  bound the empirical Rademacher complexity $$\mathop{\mathbb{E}}\limits_{\boldsymbol{\sigma}}{\left[ \sup_{h_{\boldsymbol{\theta}}\in\mathcal{H} }\frac{1}{n}\sum_{i=1}^{n}{\sigma_{i}h_{\boldsymbol{\theta}}\left(\boldsymbol{x}_{i} \right)  }\right] }$$ in terms of covering number.
	
	First, with respect to the diamond norm, for each $j\in\mathbb{N}_0$, there exists an $2^{-j}$-covering net denoted by $\mathcal{N}_{j}$ for
	the set of  quantum channels $\left\lbrace \mathcal{E}_{\boldsymbol{\theta}}\right\rbrace $,
	satisfying $N_j\leq \big{\lceil}2^{j}\pi K\big{\rceil}^{K}$. To be specific, for each $j\in\mathbb{N}$ and every parameter setting $\boldsymbol{\theta}$, there exists a quantum operator $\mathcal{E}_{\boldsymbol{\theta},j}\in\mathcal{N}_{j}$ such that
	\begin{equation*}
		\Vert\mathcal{E}_{\boldsymbol{\theta}}-\mathcal{E}_{\boldsymbol{\theta},j}\Vert_{\diamond}\le\frac{1}{2^{j}}.
	\end{equation*}
	For $j=0$, the $1$-covering net of $\left\lbrace \mathcal{E}_{\boldsymbol{\theta}}\right\rbrace $ is $\mathcal{N}_{0}=\left\lbrace 0\right\rbrace $. Moreover, according to Lemma~\ref{TH1}, the covering number $N_{j}$ can be upper bounded by $\big{\lceil}2^{j}\pi K\big{\rceil}^{K}$ .
	
	Then, for any $k\in\mathbb{N}$, we have
	\begin{equation*}
		\mathcal{E}_{\boldsymbol{\theta}}=\mathcal{E}_{\boldsymbol{\theta}}
		-\mathcal{E}_{\boldsymbol{\theta},k}+\sum_{j=k}^{1}{\mathcal{E}_{\boldsymbol{\theta},j}-\mathcal{E}_{\boldsymbol{\theta},j-1}},
	\end{equation*}
	and
	\begin{align}
		\mathop{\mathbb{E}}\limits_{\boldsymbol{\sigma}}{\left[ \sup_{h_{\boldsymbol{\theta}}\in\mathcal{H} }\frac{1}{n}\sum_{i=1}^{n}{\sigma_{i}h_{\boldsymbol{\theta}}\left(\boldsymbol{x}_{i} \right)  }\right] }
		=&\mathop{\mathbb{E}}\limits_{\boldsymbol{\sigma}}{\left[ \sup_{h_{\boldsymbol{\theta}}\in\mathcal{H} }\frac{1}{n}\sum_{i=1}^{n}{\sigma_{i}\cdot{\rm sgn}({\rm Tr}\left[ Z \mathcal{E}_{\boldsymbol{\theta}}\left( \rho\left( \boldsymbol{x}\right)\right) \right])  }\right] }\nonumber
		\\
		=&\mathop{\mathbb{E}}\limits_{\boldsymbol{\sigma}}{\left[ \sup_{h_{\boldsymbol{\theta}}\in\mathcal{H} }\frac{1}{n}\sum_{i=1}^{n}{\sigma_{i}{\rm Tr} \left[  Z \mathcal{E}_{\boldsymbol{\theta}}\left( \rho\left( \boldsymbol{x}_i\right) \right)  \right]  }\right] }\label{remove_sgn}
		\\
		\le& \frac{1}{2^k}+\frac{6}{\sqrt{n}}\sum_{j=1}^{k}{\frac{1}{2^j}\sqrt{\log N_{j}}}\label{lemmab2}
		\\
		\le& \frac{1}{2^k}+\frac{6}{\sqrt{n}}\sum_{j=1}^{k}{\frac{1}{2^j}\sqrt{K\log \lceil 2^j\pi K\rceil}}
		\\
		\le& \frac{1}{2^k}+12\sqrt{\frac{K}{n}}\int_{\frac{1}{2^{\left(k+1\right)}}}^{\frac{1}{2}}{\sqrt{\log\bigg{\lceil} \frac{\pi K}{\alpha}\bigg{\rceil}}}{\rm d}\alpha,
		\\
		\le& \frac{1}{2^k}+12\sqrt{\frac{K}{n}}\int_{\frac{1}{2^{\left(k+1\right)}}}^{\frac{1}{2}}{\sqrt{\log\frac{7K}{2\alpha}}}{\rm d}\alpha,
	\end{align}
	where Eq.~(\ref{lemmab2}) holds due to the expectation over $\boldsymbol{\sigma}$, and Eq.~(\ref{lemmab2}) is derived following a similar analysis as that in Supplementary Information of \cite{caro2022gen}.
	
	By taking the limit $k\rightarrow\infty$, we obtain
	\begin{align}
		\mathop{\mathbb{E}}\limits_{\boldsymbol{\sigma}}{\left[ \sup_{h_{\boldsymbol{\theta}}\in\mathcal{H} }\frac{1}{n}\sum_{i=1}^{n}{\sigma_{i}h_{\boldsymbol{\theta}}\left(\boldsymbol{x}_{i} \right)  }\right] }
		\le\,& 12\sqrt{\frac{K}{n}}\int_{0}^{\frac{1}{2}}{\sqrt{\log\frac{7K}{2\alpha}}}{\rm d}\alpha
		\\
		=\,& 12\sqrt{\frac{K}{n}}\left[\frac{\sqrt{\log 7K}}{2}+\frac{7\sqrt{\pi}K}{4}{\rm erfc}\left(\sqrt{\log 7K} \right) \right] \label{radhbb},
		\\
		\le\,&6\sqrt{\frac{K\log 7K}{n}}+2\sqrt{\frac{K}{n}} \label{radhbbb},
	\end{align}
	where Eq.~(\ref{radhbb}) is derived by using the integral
	\begin{equation}\label{key}
	\int{\sqrt{\log\frac{1}{\alpha}}}{\rm d}\alpha=\alpha\sqrt{\log\frac{1}{\alpha}}-\frac{\sqrt{\pi}}{2}{\rm erf}\left(\sqrt{\log\frac{1}{\alpha}} \right),
	\end{equation}
	with the error function defined by ${\rm erf}\left(x \right)=\frac{2}{\sqrt{\pi}}\int_{0}^{x}{e^{-t^2}}{\rm d}t$ and the complementary error function  ${\rm erfc}\left(x \right)=1-{\rm erf}\left(x \right)$.
	Besides, we notice that the function $K{\rm erfc}\left(\sqrt{\log 7K} \right)$ decreases monotonically with $K\in \mathbb{N}$ and is upper bounded by ${\rm erfc}\left(\sqrt{\log 7} \right) \le 0.0486$, so we have Eq.~(\ref{radhbbb}).	
	
	Thus, by taking expectations of both sides of Eq.~(\ref{radhbbb}) over the training sample set  $S$, we have Eq.~(\ref{radhb}).
\end{proof}

\twocolumngrid
\bibliography{ref}

\end{document}